\documentclass[12pt]{iopart}

\usepackage{iopams}
\usepackage{graphicx}
\usepackage{cite}
\usepackage{multirow}

\newtheorem{definition}{Definition}
\newtheorem{proposition}[definition]{Proposition}
\newtheorem{lemma}[definition]{Lemma}

\newtheorem{theorem}[definition]{Theorem}
\newtheorem{corollary}[definition]{Corollary}

\def\squareforqed{\hbox{\rlap{$\sqcap$}$\sqcup$}}
\def\qed{\ifmmode\squareforqed\else{\unskip\nobreak\hfil
\penalty50\hskip1em\null\nobreak\hfil\squareforqed
\parfillskip=0pt\finalhyphendemerits=0\endgraf}\fi}
\def\endenv{\ifmmode\;\else{\unskip\nobreak\hfil
\penalty50\hskip1em\null\nobreak\hfil\;
\parfillskip=0pt\finalhyphendemerits=0\endgraf}\fi}
\newenvironment{proof}{\noindent \textbf{{Proof.~} }}{\qed}

\def\l{\lambda}

\def\r{\rho}
\def\s{\sigma}

\def\ph{\varphi}

\def\ps{\psi}

\def\L{\Lambda}

\newcommand{\nc}{\newcommand}
\nc{\ox}{\otimes}
\nc{\bra}[1]{\langle#1|}
\nc{\ket}[1]{|#1\rangle}
\nc{\proj}[1]{| #1\rangle\!\langle #1 |}

\begin{document}

\title[Additivity and non-additivity]{Additivity and non-additivity of multipartite entanglement measures}

\author{Huangjun Zhu}
\address{Centre for Quantum Technologies, %
National University of Singapore, Singapore 117543, Singapore}
\address{NUS Graduate School for Integrative
Sciences and Engineering, Singapore 117597, Singapore}
\ead{zhuhuangjun@nus.edu.sg}
\author{Lin Chen}
\address{Centre for Quantum Technologies, %
National University of Singapore, Singapore 117543, Singapore}
\ead{cqtcl@nus.edu.sg (Corresponding~Author)}
\author{Masahito Hayashi}
\address{Graduate School of Information Sciences, Tohoku
University, Aoba-ku, Sendai, 980-8579, Japan} \address{Centre for Quantum Technologies, %
National University of Singapore, Singapore 117543, Singapore}
\ead{hayashi@math.is.tohoku.ac.jp}

\begin{abstract}
We study the additivity property of three multipartite entanglement
measures, i.e. the geometric measure of entanglement (GM), the
relative entropy of entanglement and the logarithmic global
robustness. First, we show the additivity of GM of multipartite
states with real and non-negative entries in the computational
basis. Many states of experimental and theoretical interests have
this property, e.g. Bell diagonal states, maximally correlated
generalized Bell diagonal states, generalized Dicke states, the
Smolin state, and the generalization of D\"{u}r's multipartite bound
entangled states. We also prove the additivity of other two measures
for some of these examples. Second, we show the non-additivity of GM
of all antisymmetric states of three or more parties, and provide a
unified explanation of the non-additivity of the three measures of
the antisymmetric projector states. In particular, we derive
analytical formulae of the three measures of one copy and two copies
of the antisymmetric projector states respectively. Third, we show,
with a statistical approach, that almost all multipartite pure
states with sufficiently large number of parties are nearly
maximally entangled with respect to GM and relative entropy of
entanglement. However, their GM is not strong additive; what's more
surprising, for generic pure states with real entries in the
computational basis, GM of one copy and two copies, respectively,
are almost equal. Hence, more states may be suitable for universal
quantum computation, if measurements can be performed on two copies
of the resource states. We also show that almost all multipartite
pure states cannot be produced reversibly with the combination
multipartite GHZ states under asymptotic LOCC, unless relative
entropy of entanglement is non-additive for generic multipartite
pure states.
\end{abstract}




\vspace{2pc} \noindent{\it Keywords}: multipartite entanglement,
additivity, relative entropy of entanglement, geometric measure of
entanglement, global robustness, symmetric state, antisymmetric
state

\maketitle


\tableofcontents

\section{\label{sec:Int} Introduction}

Quantum entanglement has attracted intensive attention due to its
intriguing properties and potential applications in quantum
information processing \cite{PV,HHH}, \cite{Ha} (Chapter 8). Some
geometrically motivated entanglement measures have been providing us
new insights on quantum entanglement, e.g. entanglement of formation
\cite{bdsw96}, relative entropy of entanglement (REE)
\cite{vprk97,vp98}, geometric measure of entanglement (GM)
\cite{Shi,wg03}, global robustness (GR) \cite{VT99,HN03}, and
squashed entanglement \cite{CW}. Besides providing a simple
geometric picture, they are closely related to some operationally
motivated entanglement measures, e.g. entanglement of distillation
\cite{bdsw96} and entanglement cost \cite{Hayden}. Their additivity
property for the bipartite case has been studied by many researchers
as a central issue in quantum information theory, because this
property is closely related to operational meanings
\cite{Has,MSW,Shor,rai99,aejpvm01,vw01,mi08,CW}. However, concerning
the multipartite setting, only the additivity of squashed
entanglement has been proved \cite{YHHHOS}, while the additivity
problem on other measures has largely remained open.

In this paper, we focus on the additivity property of  three main
entanglement measures in the multipartite case, i.e. REE, GM, and
logarithmic global robustness (LGR). These entanglement measures and
their additivity property are closely related to operational
concepts in the multipartite case as mentioned below. Our results
may improve the understanding on multipartite entanglement and
stimulate more research work on the three entanglement measures as
well as others, such as the tangle \cite{ckw00} and generalized
concurrence \cite{mkb05}.

 REE is a lower bound for entanglement of formation and an
upper bound for entanglement of distillation in the bipartite case.
It has a clear statistical meaning as the minimal error rate of
mistaking an entangled state for a closest separable state
\cite{vprk97,vp98}. It has also  been employed by Linden {\it et al}
\cite{lpsw99} to study the conditions for reversible state
transformation, and by Ac\'in {\it et al} \cite{AVC} to study the
structure of reversible entanglement generating sets ~\cite{bpr00}
in the tripartite scenario. In addition, Brand\~ao and Plenio
\cite{bp08} have shown that the asymptotic REE equals an asymptotic
smooth modification of LGR and a modified version of entanglement of
distillation and entanglement cost, which means that the asymptotic
REE quantifies the entanglement resources under asymptotic
non-entangling operations. In condensed matter physics, REE is also
useful for characterizing multipartite thermal correlations
\cite{Ved04a} and macroscopic entanglement, such as that in
high-temperature superconductors \cite{Ved04b}.

GM is closely related to the construction of optimal entanglement
witnesses \cite{wg03}, and discrimination of quantum states under
LOCC \cite{hmm06, hmm08, mmv07}.  GM of tripartite pure states is
closely related to the maximum output purity of the quantum channels
corresponding to these states \cite{wh02}. Recently, GM has been
utilized to determine the universality of resource states for
one-way quantum computation \cite{NDMB07,MPMNDB10}. It has also been
applied to show that most entangled states are too entangled to be
useful as computational resources \cite{GFE09}. Furthermore, the
connection between GM defined via the convex roof and a distance
like measure has also been pointed out \cite{skb10}. In condensed
matter physics, GM is useful for studying quantum many-body systems,
such as characterizing ground state properties and detecting phase
transitions \cite{odv08,orus08, ow09}.

GR is closely related to state discrimination under LOCC
\cite{hmm06, hmm08, mmv07} and entanglement quantification with
witness operators \cite{Bra2}. It is best suited to study the
survival of entanglement in  thermal states, and to determine the
noise thresholds in the generation of resource states for
measurement-based quantum computation \cite{MAVMM}.

On the other hand, the additivity property of the three measures
REE, GM and LGR greatly affect the utility of multipartite states.
For example, in state discrimination under LOCC \cite{hmm06,hmm08},
the additivity property of these measures may affect the advantage
offered by joint measurements on multiple copies of input states
over separate measurements. The additivity property of GM of generic
multipartite states is closely related to their universality as
resource states for one-way quantum computation, as we shall see in
\sref{sec:implication}.

The additivity property of the three measures REE, GM and LGR is
also  closely related to the calculation of their asymptotic or
regularized entanglement measures, which are the asymptotic limits
of the regularized quantities with the $n$-copy state. These
asymptotic measures  will be referred to as asymptotic GM, REE and
LGR, and are abbreviated to AGM, AREE and ALGR respectively. They
are useful in the study of classical capacity of quantum
multi-terminal channels \cite{mmv07}. The AREE can be used as an
invariant when we build the minimal reversible entanglement
generating set (MREGS) under asymptotic LOCC. The MREGS is a finite
set of pure entangled states from which all pure entangled states
can be produced reversibly in the asymptotic sense, which is an
essential open problem in quantum information theory
\cite{bpr00,AVC}. The AREE also determines the rate of state
transformation under asymptotic non-entangling operations
\cite{bp08}. In the bipartite case, the AREE provides a lower bound
for entanglement cost and an upper bound for entanglement of
distillation. So it is essential to compute the regularized
entanglement measures. However, the problem is generally very
difficult. One main approach for computing these asymptotic measures
is to prove their additivity, which is another focus of the present
paper. In this case, the asymptotic measures  equal to the
respective one-shot measures.

Our main approach is the following. Under some group theoretical
conditions, Hayashi {\it et al} \cite{hmm08} showed a relation among
REE, LGR and GM. Due to this relation, we can treat the additivity
problem of REE and LGR from that of GM in this special case. Hence
we can concentrate on the additivity problem of GM.

First, we derive a novel and general  additivity theorem for GM of
multipartite states with real and non-negative entries in the
computational basis. Applying this theorem, we show the additivity
of GM of  many multipartite states of either practical or
theoretical interests, such as (1) two-qubit Bell diagonal states;
(2) maximally correlated generalized Bell diagonal states, which is
closely related to local copying \cite{Owari}; (3) isotropic states,
which is closely related to depolarization channel \cite{hh99}; (4)
generalized Dicke states \cite{dicke54}, which is useful for quantum
communication and quantum networking, and can already be realized
using current technologies \cite{kstsw07,pct09,hcr09,iil10}; (5) the
Smolin state \cite{Smolin01}, which is useful for remote information
concentration \cite{mv01}, super activation \cite{sst03} and quantum
secret sharing \cite{ah06} etc; and (6) D\"ur's multipartite
entangled states, which include bound entangled states that can
violate the Bell inequality \cite{dur01}. By means of  the relation
among the three measures GM, REE and LGR, we also show the
additivity of REE of these examples, and the additivity of LGR of
the generalized Dicke states and the Smolin state. As a direct
application, we obtain  AGM and AREE of the above-mentioned
examples, and ALGR of the generalized Dicke states and the Smolin
state.

Our approach is also able to provide a lower bound for AREE and ALGR
for generic multipartite states with non-negative entries in the
computational basis, such as isotropic states \cite{hh99}, mixtures
of generalized Dicke states. In the bipartite scenario, our lower
bound for AREE is also a lower bound for entanglement cost. For
non-negative tripartite pure states, the additivity of GM implies
the multiplicativity of the maximum output purity of the quantum
channels related to these states according to the Werner-Holevo
recipe \cite{wh02}.

Second, we show the non-additivity of GM of antisymmetric states
shared over three or more parties and many bipartite antisymmetric
states. We also quantify how additivity of GM is violated in the
case of antisymmetric projector states, which include antisymmetric
basis states and antisymmetric Werner states as special examples,
and treat the same problem for REE and LGR.   For the antisymmetric
projector states, while the three one-shot entanglement measures are
generally non-additive, we obtain a relation among AREE, AGM and
ALGR. Generalized antisymmetric states \cite{bra03} are also treated
as further counterexamples to the additivity of GM.

Third, we show, with a statistical approach, that almost all
multipartite pure states are nearly maximally entangled with respect
to GM and REE. However their GM is not strong additive; what's more
surprising, for generic pure states with real entries in the
computational basis, GM of one copy and two copies, respectively,
are almost equal. Our discovery has a great implication for the
universality of resource states for one-way quantum computation, and
for asymptotic state transformation. As a twist to the assertion of
Gross {\it et al} \cite{GFE09} that most quantum states are too
entangled to be useful as computational resources, we show that more
states may be suitable for universal quantum computation, if
measurements can be performed on two copies of the resource states.
In addition, we show that almost all multipartite pure states cannot
be prepared reversibly with multipartite GHZ states (with various
numbers of parties) under LOCC even in the asymptotic sense, unless
REE is non-additive for generic multipartite pure states.

For the convenience of the readers, we  summarize the main results
on GM, REE and LGR of the states studied in this paper in
\tref{tab:Add} and \tref{tab:NonAdd}. More details can be found in
the relevant sections of the main text.

\setlength{\tabcolsep}{0.4ex}
\begin{table}
    \caption{\label{tab:Add}
{\bf Additive cases}: GM, REE and LGR of multipartite entangled
states. All states listed here have additive GM and REE, generalized
Dicke states and the Smolin state also have additive LGR, hence the
asymptotic measures are equal to the one-short measures in these
cases. REE  of the Bell diagonal states was calculated in
\cite{vprk97, rai99}. REE   of the maximally correlated generalized
Bell diagonal states and isotropic states as well as their
additivity were obtained in \cite{rai99}. GM of the generalized
Dicke states was calculated in \cite{wg03}, REE and LGR of the
generalized Dicke states were calculated in \cite{weg04, wei08,
hmm08}. REE of the Smolin state was calculated in \cite{mv01,wag04},
REE of the D\"ur's multipartite entangled states was calculated in
\cite{wag04, wei08}.} \centering
  \begin{tabular}{ c c c c c c c }
    \br
    states             & symbol                      & GM        & REE & LGR  \\
    \hline
    Bell diagonal states &   ${\rho_{\mathrm{BD}}(p) \atop p_0\geq\frac{1}{2}}$ & $1-\log(p_0+p_1)$& $1-H(p_0,1-p_0)$&   $\log(2 p_0)$  \\
    [1ex]\hline
\begin{tabular}{c}
  maximally correlated  \\
 generalized  Bell \\diagonal states
\end{tabular}
 & $\rho_{\mathrm{MCB}}(p)$    & $\log d$    & $\log d-H(p)$   &   --\\
[0.5ex]\hline
   isotropic states    & ${\rho_{\mathrm{I},\lambda}\atop \frac{1}{d}\leq \lambda\leq 1}$& $\log\frac{d(d+1)}{\lambda d+1}$ &
   \begin{math}\begin{array}{l}
     \log d+\lambda\log \lambda\\
+(1-\lambda)\log\frac{1-\lambda}{d-1}
   \end{array}\end{math}
   &  $\log(d\lambda)$ \\
  [0.5ex]\hline
\begin{tabular}{c}
    generalized  \\
  Dicke states \\
\end{tabular}
 & $|N,\vec{k}\rangle $      &
$\log\Big[\frac{\prod_{j=0}^{d-1}
k_j!}{N!}\prod_{j=0}^{d-1}\bigl(\frac{N}{k_j}\bigr)^{k_j}\Big]$       & GM  & GM   \\
                           [0.5ex]\hline
    Smolin state       & $\rho_{ABCD} $                 & 3          & 1 & 1  \\
[0.5ex]\hline
\begin{tabular}{c}
 D\"ur's multipartite \\
 entangled states \\
\end{tabular}      & ${\rho_N(x)\atop N\geq4}$& $\begin{array}{cc}
                            \log\frac{2N}{1-x} &0\leq x\leq \frac{1}{N+1}\\
                            \log\frac{2}{x}& \frac{1}{N+1}\leq x\leq 1
                          \end{array}$                       & $x$ & -- \\
    \br
  \end{tabular}
\end{table}

 \setlength{\tabcolsep}{1ex}
\begin{table}
\caption{\label{tab:NonAdd}
{\bf Non-additive cases}:  GM, REE, and LGR of some antisymmetric
and generalized antisymmetric
  states.  All states listed  satisfy $E_{\mathrm{R}}(\rho)=R_{\mathrm{L}}(\rho)=G(\rho)-S(\rho)$, except two copies of generalized antisymmetric
  states, where it is not known. When $N=2$, the antisymmetric projector state reduces to
  the antisymmetric Werner state.
  GM of single copy of the antisymmetric basis state and generalized antisymmetric  state
  was calculated in \cite{bra03, hmm08}. REE and LGR of single copy of the antisymmetric basis state and generalized antisymmetric state
  were calculated in \cite{weg04, wei08, hmm08}.  }
  \centering
  \begin{tabular}{c c c c c c c }
    \br
    states                  & symbol                      & GM             & REE        & LGR
    \\ [1ex]
\hline
  \multirow{2}{26ex}{antisymmetric basis state (Slater determinant state)} &   $|\psi_{N-}\rangle$         & $\log N! $      &  $\log N!$    & $\log N!$   \\
    &   $|\psi_{N-}\rangle^{\otimes 2}$         & $N\log N $      &  $N\log N$    & $N\log N$
    \\ [1ex]\hline

  \multirow{3}{20ex}{antisymmetric projector state}  &   $\rho_{d,N}$             & $\log\frac{d!}{(d-N)!}  $    &  $\log N!$           &$\log N! $
  \\[0.5ex]
                  &   $\rho_{d,N}^{\otimes 2}$            & $\log\frac{ d^Nd!}{N!(d-N)!}$  &  $\log\frac{ d^NN!(d-N)!}{d!}$ & $\log\frac{ d^NN!(d-N)!}{d!}$        \\[0.5ex]
                   &   ${\rho_{d_1,N}\otimes \rho_{d_2,N} \atop (N\leq d_1\leq d_2) } $          & $\log\frac{ d_1^Nd_2!}{N!(d_2-N)!}$  &  $\log\frac{ d_1^NN!(d_1-N)!}{d_1!}$ & $\log\frac{ d_1^NN!(d_1-N)!}{d_1!}$        \\
                  [2ex]\hline
 \multirow{2}{20ex}{ generalized
   antisymmetric state}
 & $|\psi_{d,p,d^p}\rangle$ &  $\log \bigl[(d^p) !\bigr]$ &  $\log \bigl[(d^p)!\bigr]$& $\log \bigl[(d^p)!\bigr]$      \\
            & $|\psi_{d,p,d^p}\rangle^{\otimes 2}$ &  $d^p\log d^p$ & --
            &--
            \\[0.8ex]
    \br
  \end{tabular}

\end{table}

The paper is organized as follows. \Sref{sec:pre} is devoted to
reviewing the preliminary knowledge and terminology, and to showing
 the relations among the three measures REE, GM and LGR. In
\sref{sec:additivity}, we prove a general additivity theorem for GM
of multipartite states with non-negative entries in the
computational basis, and apply it to many multipartite states, e.g.
Bell diagonal states, maximally correlated generalized Bell diagonal
states,  isotropic states, generalized Dicke states, mixtures of
Dicke states, the Smolin state, and D\"{u}r's multipartite entangled
states. Also, we treat the additivity problem of REE and LGR  of
these examples, and discuss the implications of these results for
state transformation. In \sref{sec:NonAdditivity}, we focus on the
antisymmetric subspace, and show the non-additivity of GM of states
in this subspace when there are  three or more parties. We also
establish a simple relation among the three measures for the tensor
product of antisymmetric projector states, and compute GM, REE and
LGR for one copy and two copies of antisymmetric projector states,
respectively. Generalized antisymmetric states are also treated as
further counterexamples to the additivity of GM. In
\sref{sec:NonAddG}, we show that  GM is not strong additive for
almost all multipartite pure states, and that it is non-additive for
almost all multipartite pure states with real entries in the
computational basis. We then discuss the implications  of these
results for  the universality of resource states in one-way quantum
computation and for asymptotic state transformation. We conclude
with a summary and some open problems.

\section{\label{sec:pre} Preliminary knowledge and terminology}

In this section, we  recall the definitions and basic properties of
the three main multipartite  entanglement measures, that is, the
\emph{relative entropy of entanglement}, the \emph{geometric measure
of entanglement} and the \emph{global robustness of entanglement},
and introduce the additivity problem on these entanglement measures.
We also present a few known results concerning the relations among
these measures, which will play an important role later. The impact
of permutation symmetry on GM and the connection between GM of
tripartite pure states and the maximum output purity of quantum
channels are also discussed briefly.

\subsection{Geometric measure, relative entropy and global robustness of entanglement}

Consider an $N$-partite  state $\r$ shared over the parties $A_1,
\ldots, A_N$ with joint Hilbert space $\ox^{N}_{j=1} \mathcal{H}_j$.
REE  measures the minimum distance in terms of relative entropy
between the given state $\r$ and the set of separable states, and is
defined as \cite{vp98}
\begin{equation}\label{eq:REEdef}
  E_{\mathrm{R}}(\r):=\min_{\sigma\in \mathrm{SEP}}
S(\rho|\!|\sigma),
\end{equation}
where $S(\rho|\!|\sigma)=\tr \rho( \log \rho- \log\sigma)$ is the
quantum relative entropy, and the logarithm has base 2 throughout
this paper. Here ``SEP'' denotes the set of fully separable states,
which are of the form $\s = \sum_j \s^1_j \ox \cdots \ox \s^N_j$,
such that $\s^k_j$ is a single-particle state of the $k$th party.
For a pure state $\rho=|\psi\rangle\langle\psi|$,
$E_{\mathrm{R}}(|\psi\rangle)$ is used to denote
$E_{\mathrm{R}}(\rho)$ through this paper, similarly for other
entanglement measures to be introduced. Any state $\sigma$
minimizing \eref{eq:REEdef} is  a \emph{closest separable state} of
$\rho$. As its definition involves the minimization over all
separable states, REE is known only for a few examples, such as
bipartite pure states \cite{vprk97,vp98,pv01}, Bell diagonal states
\cite{vprk97, rai99}, some two-qubit states \cite{mi08}, Werner
states \cite{aejpvm01, vw01, rai01}, maximally correlated states,
isotropic states \cite{rai99}, generalized Dicke states \cite{weg04,
wei08, hmm08}, antisymmetric basis states \cite{weg04,hmm08}, some
graph states \cite{mmv07}, the Smolin state, and D\"ur's
multipartite entangled states \cite{wag04, wei08}. A numeric method
for computing REE of bipartite states has been proposed in
\cite{vp98}.

REE with respect to the set of states with positive partial
transpose (PPT) $E_{\mathrm{R},\mathrm{PPT}}$, which is obtained by
replacing the set of separable states in \eref{eq:REEdef} with the
set of PPT states, has also received much attention
\cite{rai99,aejpvm01,AMVW02}. However, in this paper, we shall
follow the definition in \eref{eq:REEdef}.

GM measures  the closest distance in terms of overlap between a
given state and the set of separable states, or equivalently, the
set of pure product states, and is defined as \cite{wg03}
\begin{eqnarray}
  \label{eq:EE}
  \Lambda^2(\rho)&:=&
\max_{\sigma\in \mathrm{SEP}}
\mathrm{tr}(\rho\sigma)
  =
\max_{|\ph\rangle\in \mathrm{PRO}}
  \langle\ph|\rho|\ph\rangle,\\
    \label{eq:gm}
  G(\rho) &:=& -2~\log \Lambda(\rho).
\end{eqnarray}
Here ``PRO" denotes the set of fully product pure states in the
Hilbert space $\ox^{N}_{j=1} \mathcal{H}_j$. Any pure product state
maximizing \eref{eq:EE} is  a  \emph{closest product state} of
$\rho$. It should be emphasized that, for mixed states, the
 GM defined in \eref{eq:gm} is not an entanglement measure proper, and there are alternative definitions of GM
 through the convex roof construction \cite{wg03}. However, GM of $\rho$ defined in \eref{eq:gm} is closely related to GM of
 the purification of $\rho$ \cite{jung08}, and also to REE and LGR of $\rho$, as we shall see later. Meanwhile, this definition is  useful in
 the construction of optimal entanglement witnesses \cite{wg03}, and in the study of state discrimination under LOCC \cite{hmm06, hmm08}.
 Thus we shall follow the definition in \eref{eq:gm} in this paper. GM is  known  only for a few  examples too, such as bipartite pure states, GHZ
 type states, generalized Dicke states \cite{wg03}, antisymmetric basis
states \cite{bra03, hmm08},  pure symmetric three-qubit states
\cite{cxz10, twp09,tkk09}, some other pure three-qubit states
\cite{wg03,tpt08, cxz10}, and some graph states \cite{mmv07}.
Several numerical methods for computing GM of multipartite states
have been proposed in \cite{EHGC04, SB07}.

Different from the above two  entanglement measures, GR \cite{VT99,
HN03} measures how sensitive an entangled state is to the mixture of
noise, and is defined as follows,
\begin{eqnarray}
  \label{eq:GR} R_{\mathrm{g}}(\rho):=\min\Bigl\{t: t\geq0, \exists \mbox{ a state }
  \Delta, \frac{\rho+t\Delta}{1+t}\in \mathrm{SEP}\Bigr\}.
\end{eqnarray}
The logarithmic global robustness of entanglement (LGR) is defined
as
\begin{eqnarray}
  \label{eq:LGR} R_{\mathrm{L}}(\rho):=\log(1+R_{\mathrm{g}}).
\end{eqnarray}
LGR is known for even fewer examples, such as bipartite pure states
\cite{VT99, HN03}, generalized Dicke states, antisymmetric basis
states \cite{wei08, hmm08}, some graph states \cite{mmv07}. A
numerical method  for computing  LGR has  been proposed in
\cite{NOP1,NOP2}.

\subsection{Additivity problem on multipartite entanglement
measures} In quantum information processing, it is generally more
efficient to process a family of quantum states together rather than
process each one individually. In this case, entanglement measures
can still serve as invariants under reversible LOCC transformation,
provided that we consider the family of states as a whole. A
fundamental problem in entanglement theory is  whether the
entanglement of the tensor product of  states is the sum of that of
each individual. First we need to make it clear what the
entanglement of the tensor product of states means. Take two states
as an example, let $\rho$ be an $N$-partite sate shared over the
parties $A_1^1,\ldots, A_N^1$, and $\sigma$ be another $N$-partite
state shared over the parties $A_1^2,\ldots,A_N^2$, where we have
added superscripts to the names of the parties to distinguish the
two states. Now there are $2N$ parties involved in the tensor
product state $\rho\otimes \sigma$, however, in most scenarios that
we are concerned, the pair of parties $A_j^1, A_j^2$ for each
$j=1,\ldots,N $ are in the same lab, and can be taken as a single
party $A_j$. In this sense, $\rho\otimes \sigma$ can be seen as an
$N$-partite state shared over the parties $A_1,\ldots, A_N$. The
definition of any entanglement measure, such as GM, REE and LGR of
the tensor product state $\rho\otimes\sigma$ follows this convention
throughout this paper; similarly for the tensor product of more than
two states, except when stated otherwise.

A particularly important case is the entanglement of the tensor
product of multiple copies of the same state. In the limit of large
number of copies, we obtain the regularized or asymptotic
entanglement measure, which reads
\begin{eqnarray}
  E^\infty(\rho):=\lim_{n\rightarrow\infty} \frac{1}{n}E(\rho^{\otimes
  n}),
\end{eqnarray}
where $E$ is the entanglement measure under consideration. When  $E$
is taken as $E_{\mathrm{R}}, G$ and $R_{\mathrm{L}}$, respectively,
the resulting regularized measures are referred to as asymptotic REE
(AREE) $E^{\infty}_{\mathrm{R}}$, asymptotic GM (AGM) $G^{\infty}$,
and asymptotic LGR (ALGR) $R^{\infty}_{\mathrm{L}}$, respectively.

The entanglement $E$ of an $N$-partite state $\rho$ is called
\emph{additive} if $E^\infty(\rho)=E(\rho)$, and \emph{strong
additive} if the equality  $E(\rho\otimes \sigma)=E(\rho)+E(\sigma)$
holds for any $N$-partite state $\sigma$. Obviously, strong
additivity implies additivity. An entanglement measure itself is
called (strong) additive if it is (strong) additive for any state.
Similarly, the entanglement of the two states $\rho,\sigma$ is
called additive if the  equality $E(\rho\otimes
\sigma)=E(\rho)+E(\sigma)$ holds.

Historically, both GM and REE had been conjectured to be additive,
until counterexamples were found. The first counterexample to the
additivity of REE is the antisymmetric Werner state found by
Vollbrecht and Werner \cite{vw01}. The first counterexample to the
additivity of GM is the tripartite antisymmetric basis state found
by Werner and Holevo \cite{wh02}. Coincidentally, the two
counterexamples are both antisymmetric states, and the tripartite
antisymmetric basis state is exactly a purification of the
two-qutrit antisymmetric Werner state. We shall reveal the reason
behind  this coincidence in \sref{sec:NonAdditivity}.

For bipartite pure states,  REE is equal to the Von Neumann entropy
of each reduced density matrix \cite{vprk97, vp98, pv01};  GM is
equal to the logarithm of the inverse of the largest eigenvalue of
each reduced density matrix \cite{wg03}; and LGR is equal to one
half the logarithm of the trace of the positive square root of each
reduced density matrix \cite{VT99,HN03}; thus REE, GM and LGR are
all additive. GM and REE are also additive for any multipartite pure
states with generalized Schmidt decomposition, such as the GHZ
state. More generally, REE (GM, LGR) of a multipartite pure state is
additive if it is equal to the same measure under some bipartite
cut. For example, some graph states have additive REE, GM and LGR
for this reason \cite{mmv07}. In general, it is very difficult to
prove the additivity or non-additivity of  GM, REE and LGR of a
given state, or to compute AGM, AREE and ALGR. The additivity of REE
is known to hold  for a few other examples, such as  maximally
correlated states, isotropic states \cite{rai99}, two-qubit Werner
states \cite{aejpvm01, rai01}, and some other two-qubit states
\cite{rai99, mi08}. Little is known about the additivity property of
GM and LGR.

\subsection{Relations among the three measures}

There is a simple inequality among the three measures  REE, GM and
LGR \cite{hmm06, wei08},
\begin{equation}
  R_{\mathrm{L}}(\rho)\geq E_{\mathrm{R}}(\rho)\geq G(\rho)-S(\rho), \label{eq:REE-GM}
\end{equation}
where $S(\rho)$ is the von Neumann entropy. So the inequality
$R_{\mathrm{L}}(\rho)\geq E_{\mathrm{R}}(\rho)\geq G(\rho)$ holds
when $\rho$ is a pure state. The same is true if the three measures
are replaced by their respective regularized measures. This
inequality and its equality condition are crucial in translating our
results on  GM to that on REE and LGR in the later  sections.

A sufficient condition for the equality is given as lemma 9
in Appendix  C of \cite{hmm08}. For convenience, we reproduce it in the following proposition, 
\begin{proposition}
\label{ProH} Assume that a projector state $\frac{P}{\mathrm{tr}P}$
satisfies the following. There exist a compact group $H$, its
unitary representation $U$, and a product state $\ket{\varphi_N}$
such that (1) $U(g)$ is a local unitary for all $g \in H$. (2)
$U(g)PU(g)^\dagger= P$. (3) The state $|\varphi_N\rangle$ is one of
the closest product states of $P$. (4) $\int_H
U(g)\proj{\varphi_N}U(g)^\dagger \mu( \mathrm{d} g) \ge
\frac{\langle\varphi_N|P|\varphi_N\rangle}{\mathrm{tr}P} P$, where
$\mu$ is the invariant probability measure on $H$. Then,
\begin{eqnarray}
  R_{\mathrm{L}}\Bigl(\frac{P}{\mathrm{tr}P}\Bigl)= E_{\mathrm{R}}\Bigl(\frac{P}{\mathrm{tr}P}\Bigl)= G\Bigl(\frac{P}{\mathrm{tr}P}\Bigl)-\log\mathrm{tr}P. \label{eq:H1}
\end{eqnarray}
Under condition (1), conditions (2)-(4) are satisfied if (5) the
range of $P$ is an irreducible representation of $H$ whose
multiplicity is one in the representation $U$.
\end{proposition}
For example, generalized Dicke states, antisymmetric basis states
\cite{hmm08,weg04}, and some graph states \cite{mmv07,hmm08} satisfy
the conditions (1)-(4), so they satisfy \eref{eq:H1}. In this case,
if GM is additive, then both LGR and REE are additive, which follows
from proposition~\ref{ProH2} below. If in addition condition (5) is
satisfied, then LGR, REE and GM  are simultaneously additive or
simultaneously non-additive, which follows from
proposition~\ref{ProH3} below.

\begin{proposition}\label{ProH2}
Assume that two multipartite states $\rho, \sigma$ satisfy
$R_{\mathrm{L}}(\rho)=E_{\mathrm{R}}(\rho)=G(\rho)-S(\rho)$,
$R_{\mathrm{L}}(\sigma)=E_{\mathrm{R}}(\sigma)=G(\sigma)-S(\sigma)$,
and $G(\rho\otimes\sigma)=G(\rho)+ G(\sigma)$, then the following
relations hold,
\begin{eqnarray}\label{eq:H2}
  R_{\mathrm{L}}(\rho \otimes \sigma) &= R_{\mathrm{L}}(\rho)+R_{\mathrm{L}}(\sigma), \qquad
  E_{\mathrm{R}}(\rho \otimes \sigma) &= E_{\mathrm{R}}(\rho)+E_{\mathrm{R}}(\sigma).
\end{eqnarray}
\end{proposition}
\begin{proof}
\begin{eqnarray}
&&R_{\mathrm{L}}(\rho)+R_{\mathrm{L}}(\sigma)\geq
R_{\mathrm{L}}(\rho\otimes\sigma)\geq E_{\mathrm{R}}(\rho\otimes
\sigma)\geq
G(\rho\otimes\sigma)-S(\rho\otimes\sigma),\nonumber\\
&&
=G(\rho)+G(\sigma)-S(\rho)-S(\sigma)=R_{\mathrm{L}}(\rho)+R_{\mathrm{L}}(\sigma).
\end{eqnarray}
\end{proof}

\begin{proposition}\label{ProH3}
Assume that $n$ projector states $\frac{P_j}{\mathrm{tr}P_j}$ for
$j=1,\ldots,n$ satisfy conditions (1) and (5) of
proposition~\ref{ProH}, then
\begin{eqnarray}
  &&R_{\mathrm{L}}\Biggl(\bigotimes_{j=1}^n\frac{P_j}{\mathrm{tr}P_j}\Biggl)=
  E_{\mathrm{R}}\Biggl(\bigotimes_{j=1}^n\frac{P_j}{\mathrm{tr}P_j}\Biggl)
  = G\Biggl(\bigotimes_{j=1}^n\frac{P_j}{\mathrm{tr}P_j}\Biggl)-\sum_{j=1}^n\log\mathrm{tr}P_j. \label{eq:H3}
\end{eqnarray}
\end{proposition}

\begin{proof}
Let $H_j$ and $U_j$ be the group and the local unitary
representation satisfying the conditions (1) and (5) of proposition
\ref{ProH} concerning the projector state
$\frac{P_j}{\mathrm{tr}(P_j)}$ for $j=1,\ldots,n$. Define the
representation $\prod_{j=1}^n\times U_j$ of the direct product group
$\prod_{j=1}^n \times G_j$ by $\bigl(\prod_{j=1}^n\times
U_j\bigr)(g_1,\ldots,g_n):=\prod_{j=1}^n\otimes U_j(g_j)$. This
satisfies the conditions (1) and  (5) of proposition~\ref{ProH}
concerning the projector state $\frac{\bigotimes_{j=1}^n P_j}{
\prod_{j=1}^n\mathrm{tr}P_j}$, which implies \eref{eq:H3}.
\end{proof}

Next, we present two known results concerning the relation between a
given  entanglement measure of  a pure multipartite state and that
of its reduced states after tracing out one party. Let
$|\psi\rangle$ be an $N$-partite pure state, and $\rho$ one of its
($N-1$)-partite reduced states. First, Jung {\it et al}
\cite{jung08} have proved that the following equality holds:
\begin{eqnarray}
  G(\rho)&=&G(\ket{\psi}). \label{eq:ReduceGM}
\end{eqnarray}
So the additivity problem on  an $N$-partite pure state is
equivalent to that on its ($N-1$)-partite reduced states.

Second, Plenio and Vedral \cite{pv01} have proved a useful
inequality concerning REE,
\begin{eqnarray}
  E_{\mathrm{R}}(\ket{\psi})\geq E_{\mathrm{R}}(\rho)+S(\rho), \label{eq:ReduceREE}
\end{eqnarray}
which means that the reduction in entanglement is no less than the
increase in entropy due to deletion of a subsystem. If
$G(|\psi\rangle)=E_{\mathrm{R}}(|\psi\rangle)$ (this is true if, for
example, proposition~\ref{ProH} is satisfied), combining
\eref{eq:REE-GM}, \eref{eq:ReduceGM} and \eref{eq:ReduceREE}, we
obtain an interesting equality,
\begin{eqnarray}\label{eq:GMREEreduce}
G(|\psi\rangle)=G(\rho)=E_{\mathrm{R}}(|\psi\rangle)=E_{\mathrm{R}}(\rho)+S(\rho).
\end{eqnarray}
In this case, the total entanglement $E_{\mathrm{R}}(|\psi\rangle)$
is the sum of the remaining entanglement $E_{\mathrm{R}}(\rho)$
after losing a subsystem and the increase in entropy $S(\rho)$.
Moreover, if GM of $|\psi\rangle$ is additive, then GM of $\rho$,
REE of $|\psi\rangle$ and that of $\rho$ are all additive.

\subsection{Geometric measure and permutation symmetry}

Permutation symmetry plays an important role in the study of
multipartite entanglement.  A multipartite state is called
(permutation) \emph{symmetric} (\emph{antisymmetric}) if its support
is contained in the symmetric (antisymmetric) subspace, and
\emph{permutation invariant} if it is invariant under permutation of
the parties. Note that both symmetric states and antisymmetric
states are permutation invariant. Hayashi {\it et al} \cite{HMMOV09}
and Wei {\it et al} \cite{ws} have shown that the closest product
state to a  symmetric pure state with non-negative amplitudes in the
computational basis can be chosen to be symmetric. H\"ubener {\it et
al} \cite{hkw09}  have shown this fact for general symmetric states
(corollary 5). In addition, if $\rho$ is a pure state shared over
three or more parties, the closest product state is necessarily
symmetric (lemma 1). Here we  present a stronger result on general
symmetric states shared over three or more parties.
\begin{proposition}\label{pro:symmetric}
The closest product state to any $N$-partite pure or mixed symmetric
state with $N\geq3$ is necessarily symmetric.
\end{proposition}
\begin{proof}
Let $\rho$ be an $N$-partite symmetric state with $N\geq3$. Assume
that  $\rho$ is mixed, otherwise the proposition is already proved
as lemma 1 in \cite{hkw09}. Suppose $|\psi\rangle$ is a purification
of  $\rho$, and $|\varphi_N\rangle$ a closest product state to
$\rho$. According to theorem 1 in \cite{jung08}, there exists a
single-particle state $|a\rangle$, such that
$|\varphi_N\rangle\otimes |a\rangle$ is a closest product state to
$|\psi\rangle$; thus $|\varphi_N\rangle$ is a closest product state
to the unnormalized state $\langle a|\psi\rangle$. Since the
purification has the form $\ket{\ps} = \sum_j
\ket{\ps_j}\otimes\ket{j}$ with each $\ket{\ps_j}$ a symmetric
$N$-partite state,  $\langle a|\psi\rangle$ is an unnormalized
$N$-partite pure symmetric state with $N\geq3$. According to lemma 1
in \cite{hkw09}, $|\varphi_N\rangle$ is necessarily symmetric too.
\end{proof}

We shall prove an analog of proposition \ref{pro:symmetric} for
antisymmetric states in \sref{sec:ass1}.

\subsection{Geometric measure of tripartite pure states and maximum output purity of quantum channels}

Finally, we mention a interesting connection between GM of
tripartite pure states and the maximum output purity of quantum
channels established by Werner and Holevo \cite{wh02}.  Let $\Phi$
be a CP map with the Kraus form $\Phi(\rho)=\sum_k A_k\rho
A_k^\dag$. The maximum output purity of the map $\Phi$ is defined as
\begin{eqnarray}
\nu_p(\Phi):=\max_{\rho} \|\Phi(\rho)\|_p,
\end{eqnarray}
where $|\!|\rho|\!|_p=(\mathrm{tr} \rho^p)^{1/p}$, and the maximum
is taken over all quantum states. From the Kraus representation of
the map $\Phi$, one can construct a tripartite state $|\Phi\rangle$
(not necessarily normalized) with components $\langle
h_j|A_k|e_l\rangle$ and vice versa, where $|h_j\rangle$s and
$|e_j\rangle$s are orthonormal bases in the appropriate Hilbert
spaces, respectively. Note that, as far as entanglement measures are
concerned, it does not matter which Kraus representation of the map
$\Phi$ is chosen, because different representations  lead  to
tripartite states which are equivalent under local unitary
transformations. It should be emphasized that the map constructed
from a generic tripartite pure state according to the above
correspondence may not be trace preserving.

The maximum output purity of the channel $\Phi$ and GM of the
tripartite state $|\Phi\rangle$ is related to each other through the
following simple formula \cite{wh02}:
\begin{eqnarray}
\nu_\infty(\Phi)=\Lambda^2(|\Phi\rangle).
\end{eqnarray}
According to this result, we can get GM of a tripartite pure state
by computing the maximum output purity $\nu_\infty$ of the
corresponding map and vice versa. Generally speaking, the
computation of the maximum output purity involves far fewer
optimization parameters. Moreover, we can  translate the
multiplicativity property about the maximum output purity to the
additivity property about GM and vice versa. Actually, the
non-additivity of GM of the tripartite antisymmetric basis state
corresponds exactly to the non-multiplicativity of the maximum
output purity $\nu_\infty$ of the Werner-Holevo channel \cite{wh02}.

\section{\label{sec:additivity}
Additivity of geometric measure of non-negative multipartite states}

A density matrix is called \emph{non-negative} if all its entries in
the computational basis are non-negative.  Many states of either
theoretical or  practical interests can be written as non-negative
states, with an appropriate choice of basis, such as (1) two-qubit
Bell diagonal states, (2) maximally correlated generalized Bell
diagonal states,  (3) isotropic states, (4) generalized Dicke
states, (5)  the Smolin state, (6) D\"ur's multipartite entangled
states. 

In this section, we prove a general theorem on the strong additivity
of GM of  non-negative states, and show the  additivity of REE and
LGR for many states mentioned in the last paragraph. For general
non-negative states, our additivity result on GM can provide a lower
bound for AREE and ALGR. These  results can be used to study state
discrimination under LOCC \cite{hmm06,hmm08}, and the classical
capacity of quantum multi-terminal channels \cite{mmv07}. The result
on AREE can be utilized to determine the possibility of reversible
transformation among certain multipartite states under asymptotic
LOCC, and determine the transformation rate under asymptotic
non-entangling operations. For non-negative bipartite states, our
results also provide a lower bound for entanglement of formation and
entanglement cost. For non-negative pure tripartite states, the
additivity of GM implies the multiplicativity of the maximum output
purity of the quantum channels related to these states according to
the Werner-Holevo recipe \cite{wh02}.

In \sref{sec:add-Gen}, we prove the strong additivity of GM of
arbitrary non-negative states, and provide a nontrivial lower bound
for AREE and ALGR, which translates to a lower bound for
entanglement of formation and entanglement cost in the bipartite
case. In \sref{sec:add-Bip}, we prove the strong additivity of GM of
Bell diagonal states, maximally correlated generalized Bell diagonal
states, isotropic states, and the additivity of REE of Bell diagonal
states, maximally correlated generalized Bell diagonal states. In
\sref{sec:add-Dic}, we prove the strong additivity of GM and
additivity of REE of generalized Dicke states and their reduced
states after tracing out one party, as well as  the additivity of
LGR of generalized Dicke states. The implications of these results
for asymptotic state transformation are also discussed briefly. In
\sref{sec:add-Mix}, we give a lower bound for AREE of  mixtures of
Dicke states. In \sref{sec:add-Smo}, we  prove the strong additivity
of GM, and the additivity of REE and LGR of the Smolin state. In
\sref{sec:add-Dur}, we prove the strong additivity of GM and
additivity of REE of D\"ur's multipartite  entangled states.

\subsection{\label{sec:add-Gen} General additivity theorem for geometric measure of non-negative states}

We start by proving our main theorem of this section.

\begin{theorem}
  \label{thm:additivityGnon-negative2}
  GM of any non-negative $N$-partite state $\rho$ is strong additive; that
  is, for any other $N$-partite state $\sigma$, the following
  equalities hold:  $\Lambda(\rho \ox \sigma) = \Lambda(\rho) \Lambda(\sigma)$,
  $G(\rho \ox \sigma) = G(\rho) +G(\sigma)$.
\end{theorem}
\begin{proof}
Assume that $|\varphi_N\rangle$ is a closest product state to
$\rho\otimes\sigma$, we can write it in the following form:
\begin{eqnarray}
&|\varphi_N\rangle=\bigotimes_{l=1}^N|a_l\rangle_{A_l^1,A_l^2}=\bigotimes_{l=1}^N
\sum_{j_l=0}^{d_l-1}
a_{lj_l}|j_l\rangle_{A_l^1}|c_{lj_l}\rangle_{A_l^2},
\end{eqnarray}
where $|j_l\rangle_{A_l^1}$s for given $l$ form an orthonormal
basis, $|c_{lj_l}\rangle_{A_l^2}$s are normalized states, and
$a_{lj_l}\geq0, \quad \sum_{j_l=0}^{d_l-1} a_{lj_l}^2=1$.
    \begin{eqnarray}
  \label{al:GMmax3}
\fl  &&\Lambda^2(\rho\otimes\sigma)=\langle\varphi_N|\rho\otimes\sigma|\varphi_N\rangle,\nonumber\\
\fl  &&= \left|
  \sum^{}_{k_1,j_1,\ldots,k_N,j_N} \left\{ \Biggl(\prod_{l=1}^N
  a_{l,j_l}a_{l,k_l}\Biggr)
    \Biggl[\Biggl(\bigotimes^N_{l=1} \bra{k_l}\Biggr)\rho\Biggl(\bigotimes^{N}_{l=1}
  \ket{j_l}\Biggr)\Biggr]
  \Biggl[\Biggl(\bigotimes^{N}_{l=1} \bra{c_{l,k_l}} \Biggr)\sigma\Biggl(\bigotimes^{N}_{l=1}
  \ket{c_{l,j_l}}\Biggr)\Biggr]\right\}
  \right |,
  \nonumber\\
\fl  &&\leq
  \sum^{}_{k_1,j_1,\ldots,k_N,j_N}  \left\{ \Biggl(\prod_{l=1}^N
  a_{l,j_l}a_{l,k_l}\Biggr)
  \Biggl[\Biggl(\bigotimes^N_{l=1} \bra{k_l}\Biggr)\rho\Biggl(\bigotimes^{N}_{l=1}
  \ket{j_l}\Biggr)\Biggr]\right\}
  \Lambda^2(\sigma)
\leq
  \Lambda^2(\rho)
  \Lambda^2(\sigma).
  \end{eqnarray}
In the above derivation, the next to last inequality is due to the
assumption that $\rho$ is non-negative, and the following
inequality:
\begin{eqnarray}\label{eq:Schwarz}
\fl \left|\Biggl(\bigotimes^{N}_{l=1} \bra{c_{l,k_l}}
\Biggr)\sigma\Biggl(\bigotimes^{N}_{l=1}
\ket{c_{l,j_l}}\Biggr)\right| \leq \left[\Biggl(\bigotimes^{N}_{l=1}
\bra{c_{l,j_l}} \Biggr)\sigma\Biggl(\bigotimes^{N}_{l=1}
\ket{c_{l,j_l}}\Biggr) \Biggl(\bigotimes^{N}_{l=1} \bra{c_{l,k_l}}
\Biggr)\sigma\Biggl(\bigotimes^{N}_{l=1}
\ket{c_{l,k_l}}\Biggr)\right]^{1/2} \nonumber\\
\fl \leq \Lambda^2(\sigma),
\end{eqnarray}
which follows from the Schwarz inequality and the definition of
$\Lambda^2(\sigma)$ \footnote{Tzu-Chieh Wei showed an alternative
proof of the  inequality in \eref{eq:Schwarz} in his comment to our
manuscript (private communication).}. Combining with the opposite
inequality $\Lambda^2(\rho\otimes\sigma)\geq\Lambda^2(\rho)
  \Lambda^2(\sigma)$, which follows from the subadditivity of GM, we conclude that $\Lambda^2(\rho\otimes\sigma)=\Lambda^2(\rho)
  \Lambda^2(\sigma)$, $G(\rho \ox \sigma) = G(\rho) +G(\sigma)$. Evidently, the
   closest product state to $\rho\otimes\sigma$
    can be chosen as the tensor product of the closest product states to
   $\rho$ and $\sigma$, respectively.
\end{proof}

Theorem~\ref{thm:additivityGnon-negative2} provides a new way to
compute  GM of the tensor product of multipartite states, when  GM
of each member is  known. In particular, it enables us to calculate
AGM of non-negative states, which are a large family of multipartite
states.

\begin{corollary}
  \label{thm:AGM}
  For any non-negative state $\rho$, $G^\infty(\rho)=G(\rho)$.
\end{corollary}

For a non-negative pure tripartite state, the additivity of GM
translates immediately to the multiplicativity of the maximum output
purity $\nu_\infty$ of the corresponding quantum channel constructed
according to the Werner-Holevo recipe \cite{wh02}. Thus,
theorem~\ref{thm:additivityGnon-negative2} may also be  useful in
the study of the additivity problem concerning quantum channels.

In addition, theorem \ref{thm:additivityGnon-negative2}  gives a
lower bound for AREE and ALGR for non-negative states. This lower
bound is often nontrivial as we shall see later. According to
\eref{eq:REE-GM},   for non-negative states $\rho_j$s,
$R_{\mathrm{L}}(\otimes_j\rho_j)\geq
E_{\mathrm{R}}(\otimes_j\rho_j)\geq
G(\otimes_j\rho_j)-S(\otimes_j\rho_j)=\sum_j G(\rho_j)-\sum_j
S(\rho_j)$, where we have employed the additivity of Von Neumann
entropy. In particular, $R_{\mathrm{L}}^\infty(\rho)\geq
E_{\mathrm{R}}^\infty(\rho)\geq G(\rho)-S(\rho)$.

For a generic bipartite state $\rho$, recall that
$E_{\mathrm{F}}(\rho)\geq E_{\mathrm{R}}(\rho)$ and
$E_{\mathrm{c}}(\rho)\geq E_{\mathrm{R}}^\infty(\rho)$, where
$E_{\mathrm{F}}$ and $E_{\mathrm{c}}$ denote entanglement of
formation and entanglement cost, respectively. Therefore, when
$\rho$ is non-negative, $G(\rho)-S(\rho)$ also gives a lower bound
for entanglement of formation and entanglement cost.

\begin{theorem}
  \label{thm:AREE}
  Both ALGR and AREE of any non-negative
  state $\rho$ are  lower bounded by the difference between GM and the Von Neumann entropy of the state, $R_{\mathrm{L}}^\infty(\rho)\geq E_{\mathrm{R}}^\infty(\rho)\geq
  G(\rho)-S(\rho)$. The bound for AREE is tight if $E_{\mathrm{R}}(\rho)=
  G(\rho)-S(\rho)$, and the bound for ALGR is tight if $R_{\mathrm{L}}(\rho)=E_{\mathrm{R}}(\rho)=
  G(\rho)-S(\rho)$. If $\rho$ is a non-negative bipartite state, $G(\rho)-S(\rho)$
  is also a lower bound for entanglement of formation and entanglement cost, $E_{\mathrm{F}}(\r) \geq E_{\mathrm{c}}(\rho)\geq G(\rho)-S(\rho)$.
\end{theorem}

Next, we prove a useful lemma concerning the closest product states
of non-negative states.
\begin{lemma}
  \label{le:closest2}
  The closest product state to any non-negative  state $\rho$ can be chosen to be non-negative.
\end{lemma}
\begin{proof}
Represent $\rho$ in the computational basis,
\begin{eqnarray}
&&\rho=\sum_{k_1,j_1,\ldots,k_N,j_N}\rho_{k_1,\ldots,k_N;j_1,\ldots,j_N}|k_1,\ldots,k_N\rangle\langle
j_1,\ldots,j_N|,
 \end{eqnarray}
where $\rho_{k_1,\ldots,k_N;j_1,\ldots,j_N}\geq0$. Assume that
$|\varphi_N\rangle$ is a closest product state to $\rho$ which reads
\begin{eqnarray}
  \ket{\varphi_N}  &:= \ket{a_1} \ox \cdots \ox \ket{a_N}\quad
  \mbox{with}\quad
     \ket{a_l} &:=  (b_{0,l}, \ldots, b_{d_l-1,l})^T \quad\forall l.
\end{eqnarray}
 \begin{eqnarray}
\fl\langle\varphi_N|\rho|\varphi_N\rangle
&=&\sum_{k_1,j_1,\ldots,k_N,j_N}\rho_{k_1,\ldots,k_N;j_1,\ldots,j_N}\langle
a_1,\ldots,a_N|k_1,\ldots,k_N\rangle\langle
j_1,\ldots,j_N|a_1,\ldots,a_N\rangle,\nonumber\\
\fl&\leq&
\sum_{k_1,j_1,\ldots,k_N,j_N}\rho_{k_1,\ldots,k_N;j_1,\ldots,j_N}\prod_{l=1}^{N}
|b_{k_l,l}^*b_{j_l,l}|,
\end{eqnarray}
the inequality is saturated when  $b_{j_l, l}$s are all
non-negative, that is $|\ph_N\rangle$ is non-negative.
\end{proof}

In the rest of this section, we illustrate the power of
theorems~\ref{thm:additivityGnon-negative2}, \ref{thm:AREE} and
lemma~\ref{le:closest2} with many concrete examples. In particular,
we prove the strong additivity of GM of the following states: Bell
diagonal states, maximally correlated generalized Bell diagonal
states, isotropic states, generalized Dicke states, mixtures of
Dicke states, the Smolin state, and D\"ur's multipartite  entangled
states. Moreover, we prove the additivity of REE  of Bell diagonal
states, maximally correlated generalized Bell diagonal states,
generalized Dicke states, generalized Dicke states with one party
traced out, the Smolin state, and  D\"ur's multipartite  entangled
states. The additivity of LGR of generalized Dicke states and the
Smolin state is also shown. The implications of these results for
state transformation under asymptotic LOCC and asymptotic
non-entangling operations, respectively, are also discussed briefly.

\subsection{\label{sec:add-Bip}  Bipartite mixed states and tripartite pure states}

In the bipartite scenario, for any pure states, REE, GM and LGR can
be easily calculated and their additivity has been shown
\cite{vprk97,vp98,wg03,VT99, HN03}. Note that any bipartite pure
state is non-negative in the Schmidt basis; hence, its  GM is
strong additive according to
theorem~\ref{thm:additivityGnon-negative2}. The same is true for any
multipartite state with a generalized Schmidt decomposition.
However, even in the bipartite scenario, the calculation of REE, GM
and LGR is not so trivial for mixed states. Moreover, the additivity
problem on generic mixed states is notoriously difficult. Due to
 \eref{eq:ReduceGM}, the difficulty in GM for bipartite
mixed states is equivalent to  that for  tripartite pure states.

As one of the most simple examples of bipartite mixed states, we
focus on maximally correlated generalized Bell diagonal states.
Maximally correlated states are known as a typical example where REE
is known to be additive \cite{rai99}. By applying a suitable local
unitary transformation, any maximally correlated generalized Bell
diagonal state can be transformed into the following form,
\begin{eqnarray}
  \label{al:mcgbelldiagonal}
  \r_{\mathrm{MCB}}(p) &:= \sum^{d-1}_{k=0} p_k \proj{\Psi_k}\qquad
  \mbox{with}\qquad
  \ket{\Psi_k} := \frac{1}{\sqrt d} \sum^{d-1}_{j=0} \rme^{ \frac{2\pi \rmi k}{d} j }\ket{jj},
\end{eqnarray}
where $p=(p_0,\ldots, p_{d-1})$ is a probability distribution. It's
easy to see that $\L^2(\r_{\mathrm{MCB}}(p)) =\max_{|\ph\rangle}
\bra{\ph,\ph}\r_\mathrm{MCB}(p)\ket{\ph,\ph}\leq
\max_{|\ph\rangle,k} \langle\ph,\ph\proj{\Psi_k}\ph,\ph\rangle \leq
\frac{1}{d}$, and the upper bound is achievable  by setting
$\ket{\ph}=\ket{j}, \forall j$. In addition, the state
$\r_{\mathrm{MCB}}(p)$ can be converted into a non-negative state
via a suitable local unitary transformation, such as the
simultaneous local Fourier transformation. According to
theorem~\ref{thm:additivityGnon-negative2}, we get
\begin{proposition}
  \label{pro:AREE}
  The maximally correlated generalized Bell diagonal state in
  \eref{al:mcgbelldiagonal} has strong additive GM, and thus $G^{\infty}(\r_{\mathrm{MCB}}(p))
= G(\r_{\mathrm{MCB}}(p)) =  \log{d}$.
\end{proposition}

Let $\sigma=\sum_{j}\frac{1}{d}\ket{jj}\bra{jj}$, we have
$E_{\mathrm{R}}^{\infty}( \r_{\mathrm{MCB}}(p))\le E_{\mathrm{R}}(
\r_{\mathrm{MCB}}(p)) \le S(\r_{\mathrm{MCB}}(p)\|\sigma)= \log{d} -
H(p)$, where $H(p)$ is the Shannon entropy of the distribution $p$.
Applying the inequality \eref{eq:REE-GM} and its asymptotic version
to the maximally correlated generalized Bell diagonal state
$\r_{\mathrm{MCB}}(p)$, we obtain the additivity of REE for
$\r_{\mathrm{MCB}}(p)$:
\begin{eqnarray}
E_{\mathrm{R}}^{\infty}(  \r_{\mathrm{MCB}}(p))= E_{\mathrm{R}}(
\r_{\mathrm{MCB}}(p)) = \log{d} - H(p).
\end{eqnarray}
The same result has been obtained by Rains \cite{rai99} with a
different method.

In the two-qubit system, any rank-two Bell diagonal state, a mixture
of two orthogonal Bell states, can always be converted into the
 form in \eref{al:mcgbelldiagonal}, with a suitable local unitary transformation. So, any
two-qubit rank-two Bell diagonal state has strong additive GM and
additive REE. Actually, this is true for all Bell diagonal states.
Let $\rho_{\mathrm{BD}}$ be any Bell diagonal state,
\begin{eqnarray}\label{eq:BD}
\rho_{\mathrm{BD}}(p)=\sum_{j=0}^3p_j|\Psi_j\rangle\langle\Psi_j|,
\end{eqnarray}
where $p=(p_0,p_1,p_2,p_3)$ is a probability distribution, and
$|\Psi_j\rangle$s are the standard Bell basis. $|\Psi_0\rangle,
|\Psi_1\rangle$ are already defined in \eref{al:mcgbelldiagonal},
the other two states are defined as
\begin{eqnarray}
&&|\Psi_2\rangle:=\frac{1}{\sqrt{2}}(|01\rangle+|10\rangle),\qquad
|\Psi_3\rangle:=\frac{1}{\sqrt{2}}(|01\rangle-|10\rangle).
\end{eqnarray}
Since local unitary transformations can realize all 24 permutations
of the four Bell states, with out loss of generality, we may assume
$p_0\geq p_1\geq p_2\geq p_3$.  Then $\rho_{\mathrm{BD}}$ is clearly
a non-negative state, and its GM is strong additive according to
theorem~\ref{thm:additivityGnon-negative2}. Meanwhile,  its closest
product state can be chosen to be  non-negative according to
lemma~\ref{le:closest2}. Let
$|\varphi_2\rangle=(\cos\theta_1|0\rangle+\sin\theta_1|1\rangle)\otimes(\cos\theta_2|0\rangle+\sin\theta_2|1\rangle)$
with $0\leq \theta_1,\theta_2\leq\frac{\pi}{2}$.
\begin{eqnarray}
\fl\Lambda^2(\rho_{\mathrm{BD}}(p))&=&\max_{|\varphi_2\rangle}
\langle\varphi_2|\rho_{\mathrm{BD}}(p)|\varphi_2\rangle\nonumber\\
\fl&=&\frac{1}{2} \max_{\theta_1,\theta_2}
\Bigl[p_0\cos^2(\theta_1-\theta_2)+p_1\cos^2(\theta_1+\theta_2)
+p_2\sin^2(\theta_1+\theta_2)\nonumber\\
\fl&&{}+p_3\sin^2(\theta_1-\theta_2)\Bigr]\nonumber\\
\fl&=&\frac{1}{2} \max_{\theta_1,\theta_2}
\Bigl[p_2+p_3+(p_0-p_3)\cos^2(\theta_1-\theta_2)+(p_1-p_2)\cos^2(\theta_1+\theta_2)\Bigr]\nonumber\\
\fl&=&\frac{p_0+p_1}{2}.
\end{eqnarray}
The maximum in the above equation can be obtained at
$\theta_1=\theta_2=0$, that is $|\varphi_2\rangle=|00\rangle$.

REE of Bell diagonal states have  been computed by Vedral {\it et
al} \cite{vprk97} and by Rains \cite{rai99}, with the result,
$E_{\mathrm{R}}(\rho_{\mathrm{BD}}(p))=0$ if $p_0\leq \frac{1}{2}$,
and
\begin{eqnarray}\label{eq:BDREE}
E_{\mathrm{R}}(\rho_{\mathrm{BD}}(p))=1+p_0\log p_0+(1-p_0)\log
(1-p_0)=1-H(p_0,1-p_0)
\end{eqnarray}
if $p_0\geq \frac{1}{2}$, where $H(p_0,1-p_0)$ is the binary Shannon
entropy. Although, in general,
$E_{\mathrm{R}}(\rho_{\mathrm{BD}}(p))>G(\rho_{\mathrm{BD}}(p))-S(\rho_{\mathrm{BD}}(p))$,
except for rank-two Bell diagonal states, REE of Bell diagonal
states is also additive. This can be shown as follows, with a
suitable local unitary transformation and twirling,
$\rho_{\mathrm{BD}}(p)$ can be turned into a Werner state with the
same maximal eigenvalue $p_0$, and thus with the same REE according
to \eref{eq:BDREE}. Recall that REE of any two-qubit Werner state is
additive \cite{aejpvm01,rai01}, it follows from the monotonicity of
AREE under LOCC that REE of any Bell diagonal state is also
additive.

\begin{proposition}
The Bell diagonal state in \eref{eq:BD} has strong additive GM, and
additive REE, thus
$G^\infty(\rho_{\mathrm{BD}}(p))=G(\rho_{\mathrm{BD}}(p))=1-\log(p_0+p_1)$,
and
$E_{\mathrm{R}}^\infty(\rho_{\mathrm{BD}}(p))=E_{\mathrm{R}}(\rho_{\mathrm{BD}}(p))=1-H(p_0,1-p_0)$.
\end{proposition}

To compute LGR of the Bell diagonal state $\rho_{\mathrm{BD}}(p)$,
let $\rho^\prime$ be an unnormalized separable state with the
minimal trace such that $\rho^\prime\geq \rho_{\mathrm{BD}}$, then
$R_{\mathrm{L}}(\rho_{\mathrm{BD}}(p))=\log[\mathrm{tr}(\rho^\prime)]$.
In addition, $\rho^\prime$ can also be chosen to be a Bell diagonal
state. Since a Bell diagonal state is separable if and only if its
largest eigenvalue is no larger  than one half of its trace,
$\rho^\prime$ can be chosen to be
$\rho^\prime=\rho_{\mathrm{BD}}(p)+\frac{(2p_0-1)}{3}(I-|\Psi_0\rangle\langle\Psi_0|)$.
We thus obtain
\begin{eqnarray}
\quad R_{\mathrm{L}}(\rho_{\mathrm{BD}}(p))=\log (2p_0) \quad
\mbox{for} \quad p_0\geq \frac{1}{2}.
\end{eqnarray}

Next, we consider the isotropic state $\r_{\mathrm{I},\lambda}:=
\frac{1-\lambda}{d^2-1}(I-\ket{\Psi_0}\bra{\Psi_0})
+\lambda\ket{\Psi_0}\bra{\Psi_0}$ with $\frac{1}{d^2}\leq \lambda
\leq 1$. It is easy to see that $\L^2(\r_{\mathrm{I},\lambda}) =
\frac{\lambda d+1}{d(d+1)}$, and that the state $\ket{jj}$ for each
$j=0,1,\ldots,d-1$ is a closest product state. Since
$\r_{\mathrm{I},\lambda}$ is a non-negative state, its GM is strong
additive according to theorem~\ref{thm:additivityGnon-negative2}. So
we obtain

\begin{proposition}
  The isotropic state $\r_{\mathrm{I},\lambda}$ with $\frac{1}{d^2}\leq \lambda
\leq 1$ has strong additive GM, and
  thus
  $G^{\infty}(\r_{\mathrm{I},\lambda})= G(\r_{\mathrm{I},\lambda})=\log
  \frac{d(d+1)}{\lambda d+1} $.
\end{proposition}
The  REE and AREE of the isotropic state were calculated by Rains
\cite{rai99} with the result,
\begin{eqnarray}
\fl E_{\mathrm{R}}^{\infty}(\r_{\mathrm{I},\lambda})=
E_{\mathrm{R}}(\r_{\mathrm{I},\lambda})=\left\{\begin{array}{cc}
                                      0, &
                                    0\leq\lambda\leq\frac{1}{d},\\
\log d  + \l \log \l + (1-\l) \log \frac{1-\l}{ d-1},& \frac{1}{d}\leq\lambda\leq 1.\\
                                    \end{array}\right.
  \end{eqnarray}

To compute LGR of the isotropic state $\r_{\mathrm{I},\lambda}$, let
$\rho^\prime$ be an unnormalized separable state with the minimal
trace such that $\rho^\prime\geq \r_{\mathrm{I},\lambda}$, then
$R_{\mathrm{L}}(\r_{\mathrm{I},\lambda})=\log[\mathrm{tr}(\rho^\prime)]$.
In addition, $\rho^\prime$ can also be chosen to be an isotropic
state. Since the isotropic state $\rho_{\mathrm{I},\lambda}$ is
separable if $0\leq\lambda\leq \frac{1}{d}$, and entangled
otherwise, $\rho^\prime$ can be chosen to be
$\rho^\prime=\rho_{\mathrm{I},\lambda}+\frac{d\lambda-1}{d^2-1}(I-|\Psi_0\rangle\langle\Psi_0|)$.
We thus obtain
\begin{eqnarray}
\quad
R_{\mathrm{L}}(\rho_{\mathrm{I},\lambda})=\left\{\begin{array}{cc}
                                             0,&  0\leq \lambda\leq\frac{1}{d}, \\
                                             \log (d\lambda),&
                                             \frac{1}{d}\leq
                                             \lambda\leq1.
                                           \end{array}\right.
\end{eqnarray}

Now, we focus on  pure three-qubit states as the most simple
multipartite pure states. Recall that any pure three-qubit state can
be turned into the following form via a suitable local unitary
transformation \cite{aac00},
\begin{eqnarray}
|\psi\rangle&=&\lambda_0|000\rangle+\lambda_1\rme^{\rmi\phi}|100\rangle+\lambda_2|101\rangle+\lambda_3|110\rangle+\lambda_4|111\rangle, \nonumber\\
&&\lambda_j\geq0,  \quad\sum_j\lambda_j^2=1, \quad0\leq\phi\leq\pi.
\end{eqnarray}
If $\phi=0$, the resulting four-parameter family of states  are all
non-negative. In that case, according to
theorems~\ref{thm:additivityGnon-negative2} and \ref{thm:AREE},
their GM is strong additive and gives a lower bound for their AREE
and ALGR. The bound for AREE and ALGR is tight for the W state as we
shall see in \sref{sec:add-Dic}.

For generic two-qubit states, previous numerical calculation in
\cite{vp98} found no counterexample to the additivity of REE, while
our numerical calculation found no counterexample to the additivity
of GM. We thus conjecture that both REE and GM are additive for
generic two-qubit states. Note that each bipartite reduced state of
a pure three-qubit state is a rank-two two-qubit state. According to
\eref{eq:ReduceGM}, GM of pure three-qubit states would be additive
if GM of general two-qubit states were additive.

\subsection{\label{sec:add-Dic} Generalized Dicke states}
Generalized Dicke states are also called symmetric basis states;
they are defined in $\mathcal{H}=(\mathbb{C}^d)^{\ox N}$ as follows
\cite{wg03, dicke54},
\begin{eqnarray}
  \label{al:dicke}
  &&\ket{N, \vec{k}} :=
  \frac{1}{ \sqrt{ C_{N,\vec{k}} } }
  \sum_P P
  (
  \ket{
  {
  \overbrace{0, \ldots, 0}^{k_0},
  \overbrace{1, \ldots, 1}^{k_1},
  \ldots,
  \overbrace{d-1, \ldots, d-1}^{k_{d-1}}
  }
  }
  ), \nonumber\\
  &&\vec{k} := (k_0, k_1, \ldots, k_{d-1}), \qquad \sum^{d-1}_{j=0}
  k_j=N.
\end{eqnarray}
Here $\{P\}$ denotes the set of all distinct permutations of the
spins, and $C_{N,\vec{k}} = \frac {N!} { \prod^{d-1}_{j=0} {k_j}! }$
is the normalization factor. When $d\geq N$, the state
$|N,(\overbrace{1,\ldots,1}^{N},\overbrace{0,\ldots,0}^{d-N})\rangle$
is sometimes referred to as the totally symmetric basis state and
written as  $|\psi_{N+}\rangle$ \cite{hmm08}. When $d=2$,
$|N,(k_0,k_1)\rangle$ is called a Dicke state and denoted as
$|N,k_0\rangle$. Dicke states are useful for quantum communication
and quantum networking \cite{kstsw07,pct09}. Some typical Dicke
states have been realized in trapped atomic ions \cite{hcr09}.
Recently, the multiqubit Dicke state with half excitations $|N,
N/2\rangle$ has been employed to implement a scalable quantum search
based on Grover's algorithm by using adiabatic techniques
\cite{iil10}. In view of the fast progress made in experiments,
further theoretical study is required to explore the full potential
of Dicke states.

GM, REE and LGR of the generalized Dicke states have been computed
in \cite{wg03,wei08,hmm08} with the result,
\begin{eqnarray}
\label{eq:DickeREEGM}
  R_{\mathrm{L}}\bigl(|N,\vec{k}\rangle\bigr)= E_{\mathrm{R}}\bigl(|N,\vec{k}\rangle\bigr)=G\bigl(|N,\vec{k}\rangle\bigr)
                                  = -\log\Bigg[\frac{N!}{\prod_{j=0}^{d-1}
k_j!}\prod_{j=0}^{d-1}\Bigl(\frac{k_j}{N}\Bigr)^{k_j}\Bigg].
\end{eqnarray}
In addition, the generalized Dicke states have been proved to
satisfy the conditions (1)-(4) of proposition \ref{ProH} in section
III B of \cite{hmm08}. Since the generalized Dicke states have
non-negative amplitudes, theorem~\ref{thm:additivityGnon-negative2}
and proposition \ref{ProH2} imply that
\begin{eqnarray}
  \label{eq:additivityEsym}
   &&R_{\mathrm{L}}\Biggl(\bigotimes_\alpha |N,\vec{k}_\alpha\rangle\Biggr)
  =E_{\mathrm{R}}\Biggl(\bigotimes_\alpha |N,\vec{k}_\alpha\rangle\Biggr)
  =G\Biggl(\bigotimes_\alpha |N,\vec{k}_\alpha\rangle\Biggr) \nonumber\\
   && =\sum_{\alpha } R_{\mathrm{L}} \Bigl(|N,\vec{k}_\alpha\rangle\Bigr)
    =\sum_{\alpha } E_{\mathrm{R}}\Bigl(|N,\vec{k}_\alpha\rangle\Bigr)
  =\sum_{\alpha } G\Bigl(|N,\vec{k}_\alpha\rangle\Bigr)\nonumber\\
  &&=-\sum_\alpha\log\Bigg[\frac{N!}{\prod_{j=0}^{d_{\alpha}-1}
   k_{\alpha,j}!}\prod_{j=0}^{d_{\alpha}-1}\Bigl(\frac{k_{\alpha,j}}{N}\Bigr)^{k_{\alpha,j}}\Bigg].
\end{eqnarray}
In particular when all states $|N,\vec{k}_\alpha\rangle$ are
identical, we get
\begin{proposition}
  Generalized Dicke states have strong additive GM, additive REE and LGR,
 hence
\begin{eqnarray}
  \label{DickeAREE}
  \fl R_{\mathrm{L}}^\infty\Bigl(|N,\vec{k}\rangle\Bigr)=E_{\mathrm{R}}^\infty\Bigl(|N,\vec{k}\rangle\Bigr)=G^\infty\Bigl(|N,\vec{k}\rangle\Bigr)
    =-\log\Bigg[\frac{N!}{\prod_{j=0}^{d-1}
k_j!}\prod_{j=0}^{d-1}\Bigl(\frac{k_j}{N}\Bigr)^{k_j}\Bigg].
\end{eqnarray}
\end{proposition}

Let  $\widetilde{\rho}_{N,\vec{k}}$ be the ($N-1$)-partite reduced
state of the $N$-partite generalized Dicke state
$|N,\vec{k}\rangle$. Since
$E_{\mathrm{R}}\bigl(|N,\vec{k}\rangle\bigr)=G\bigl(|N,\vec{k}\rangle\bigr)$,
equation  \eref{eq:GMREEreduce} implies that
$E_{\mathrm{R}}\bigl(\widetilde{\rho}_{N,\vec{k}}\bigr)=E_{\mathrm{R}}\bigl(|N,\vec{k}\rangle\bigr)-S\bigl(\widetilde{\rho}_{N,\vec{k}}\bigr)
=E_{\mathrm{R}}\bigl(|N,\vec{k}\rangle\bigr)-H\bigl(\frac{\vec{k}}{N}\bigr)$,
where $H\bigl(\frac{\vec{k}}{N}\bigr)$ is the Shannon entropy. This
equality has already been proved in \cite{wei08} with explicit
calculation. In contrast, our derivation is much simpler and more
general. Finally, since REE of $\widetilde{\rho}_{N,\vec{k}}$ is
also additive, we get the AREE as follows,
\begin{eqnarray}
&&E_{\mathrm{R}}^\infty\bigl(\widetilde{\rho}_{N,\vec{k}}\bigr)=E_{\mathrm{R}}\bigl(\widetilde{\rho}_{N,\vec{k}}\bigr)=-\log\Bigg[\frac{N!}{\prod_{j=0}^{d-1}
k_j!}\prod_{j=0}^{d-1}\Bigl(\frac{k_j}{N}\Bigr)^{k_j}\Bigg]-H\Biggl(\frac{\vec{k}}{N}\Biggr).
\end{eqnarray}
In the case  $N=3$, the above result gives a lower bound for the
entanglement cost of the following two states, respectively: the
two-qubit state
$\frac13(\ket{01}+\ket{10})(\bra{01}+\bra{10})+\frac13\proj{00}$ and
the two-qutrit state $\frac16(\ket{01}+\ket{10})(\bra{01}+\bra{10})+
\frac16(\ket{02}+\ket{20})(\bra{02}+\bra{20})+
\frac16(\ket{21}+\ket{12})(\bra{21}+\bra{12})$.

Another application of our result is to help determine whether two
multipartite pure states can be inter-converted reversibly under
asymptotic LOCC, and help solve the long standing problem about
MREGS \cite{bpr00,AVC}. Consider two tripartite states
$|\psi_1\rangle, \ket{\psi_2}$ over the three parties $A_1, A_2,
A_3$. According to the result of Linden {\it et al} \cite{lpsw99},
reversible transformation between the two states under asymptotic
LOCC would mean the ratio of the AREE
$E_{\mathrm{R}}^\infty(A_1:A_2A_3)$ across the cut $A_1:A_2A_3$ to
the tripartite AREE $E_{\mathrm{R}}^\infty$ is conserved.
\Tref{tab:GHZWSAS} shows the bipartite and tripartite AREE of the
GHZ state, W state, tripartite totally symmetric and antisymmetric
basis states $|\psi_{3\pm}\rangle$ ($|\psi_{3-}\rangle$ is defined
in \eref{eq:abs} in \sref{sec:ass1}) respectively. The inequality
$E_{\mathrm{R}}^{\infty}(\psi_{3-})\geq \log 5$ in the table follows
from \eref{eq:ReduceREE} and the result
$E_{\mathrm{R}}^{\infty}(\mathrm{tr}_{A_1}|\psi_{3-}\rangle\langle\psi_{3-}|)\geq
E_{\mathrm{R},\mathrm{PPT}}^{\infty}(\mathrm{tr}_{A_1}|\psi_{3-}\rangle\langle\psi_{3-}|)=\log\frac{5}{3}$
\cite{aac00}. With these results, it is immediately clear that there
is no reversible transformation between any two states among the
four states.
\begin{table}[h]
  \centering
    \caption{\label{tab:GHZWSAS}   Bipartite (across the cut $A_1:A_2A_3$) and tripartite AREE of the  GHZ state,
   W state,   totally symmetric basis state $|\psi_{3+}\rangle=|3,(1,1,1)\rangle$ and antisymmetric basis state $|\psi_{3-}\rangle$, respectively.
   The ratio of the bipartite AREE to the tripartite AREE listed in the last column of the table
   decreases monotonically
   down the column, which implies that there is no reversible transformation  between any two of the four states under asymptotic LOCC. }
  \begin{math}
  \begin{array}{c c c c}
  \br
    \mbox{states} &E_{\mathrm{R}}^\infty(A_1:A_2A_3)  & E_{\mathrm{R}}^\infty&  \mbox{ratio} \\
    \hline
    \mathrm{GHZ} & 1 & 1 & 1  \\
    \mathrm{W} & \frac{1}{3}\log\frac{27}{4} & \log\frac{9}{4} & 0.7849  \\
    |\psi_{3+}\rangle & \log 3 & \log\frac{9}{2} &  0.7304 \\
  |\psi_{3-}\rangle & \log 3 & \geq \log 5 &\leq 0.6826\\
  \br
  \end{array}
  \end{math}

\end{table}

Similar argument can be used to show that the transformation between
the $N$-partite GHZ state and any $N$-partite symmetric basis state
is not reversible. Also, the transformation between two symmetric
basis states is  generally not reversible if they cannot be
converted into each other by a permutation of the  kets in the
computational basis.

\subsection{\label{sec:add-Mix}  Mixture of Dicke states}
Next, we consider the mixture of Dicke states
\begin{eqnarray}
\rho(\{p_{k} \}):=\sum_{k}p_{k} |N,k\rangle\langle N,k|.
\end{eqnarray}
 REE of these states has been derived by Wei
\cite{weg04, wei08}. We shall give a  lower bound for AREE of these
states based on the relation between REE and GM. Similar techniques
can also be applied to  the mixture of generalized Dicke states. The
lower bound can often be improved if the convexity of AREE is taken
into account, as we shall see shortly. For simplicity, we illustrate
our method with the mixture of two Dicke states. Following
\cite{weg04,wei08}, define
\begin{eqnarray}
\fl\rho_{N;k_1,k_2}(s):=s|N,k_1\rangle\langle
N,k_1|+(1-s)|N,k_2\rangle\langle N,k_2|, \quad 0\leq s\leq 1, \quad
k_1<k_2.
\end{eqnarray}
Since the mixture of Dicke states is both symmetrical and
non-negative,  corollary 5 in \cite{hkw09} (see also
proposition~\ref{pro:symmetric}) and lemma~\ref{le:closest2} implies
that the closest product state to $\rho_{N;k_1,k_2}(s)$ can be
chosen to be of the form
$\ket{\varphi_N}=\bigl(\cos\theta|0\rangle+\sin\theta|1\rangle\bigr)^{\otimes
N}$ with $0\leq \theta\leq \frac{\pi}{2}$.
\begin{eqnarray}
\fl
\Lambda^2(\rho_{N;k_1,k_2}(s))=\max_{\theta}\langle\varphi_N|\rho_{N;k_1,k_2}(s)|\varphi_N\rangle\nonumber
\\=\max_{\theta}
\Biggl[s{N \choose k_1}\cos^{2k_1}\theta\sin^{2N-2k_1}\theta+(1-s){N
\choose k_2}\cos^{2k_2}\theta\sin^{2N-2k_2}\theta\Biggr].
\end{eqnarray}
The maximization over $\theta$ is easy to carry out; for example,
let $x=\cos^2\theta$, the extremal condition  leads to a
$(k_2-k_1+1)$-order polynomial equation in $x$, which can be solved
straightforwardly . In particular, this equation  can be solved
analytically if $k_2-k_1\leq3$. Since $\rho_{N;k_1,k_2}(s)$ is
non-negative, according to
theorems~\ref{thm:additivityGnon-negative2} and \ref{thm:AREE},
$G\bigl(\rho_{N;k_1,k_2}(s)\bigr)$ is strong additive, and
$E_{\mathrm{R}}^\infty\bigl(\rho_{N;k_1,k_2}(s)\bigr)$ is lower
bounded by
$G\bigl(\rho_{N;k_1,k_2}(s)\bigr)-S\bigl(\rho_{N;k_1,k_2}(s)\bigr)$.

\begin{figure}
  \includegraphics[width=5.2cm]{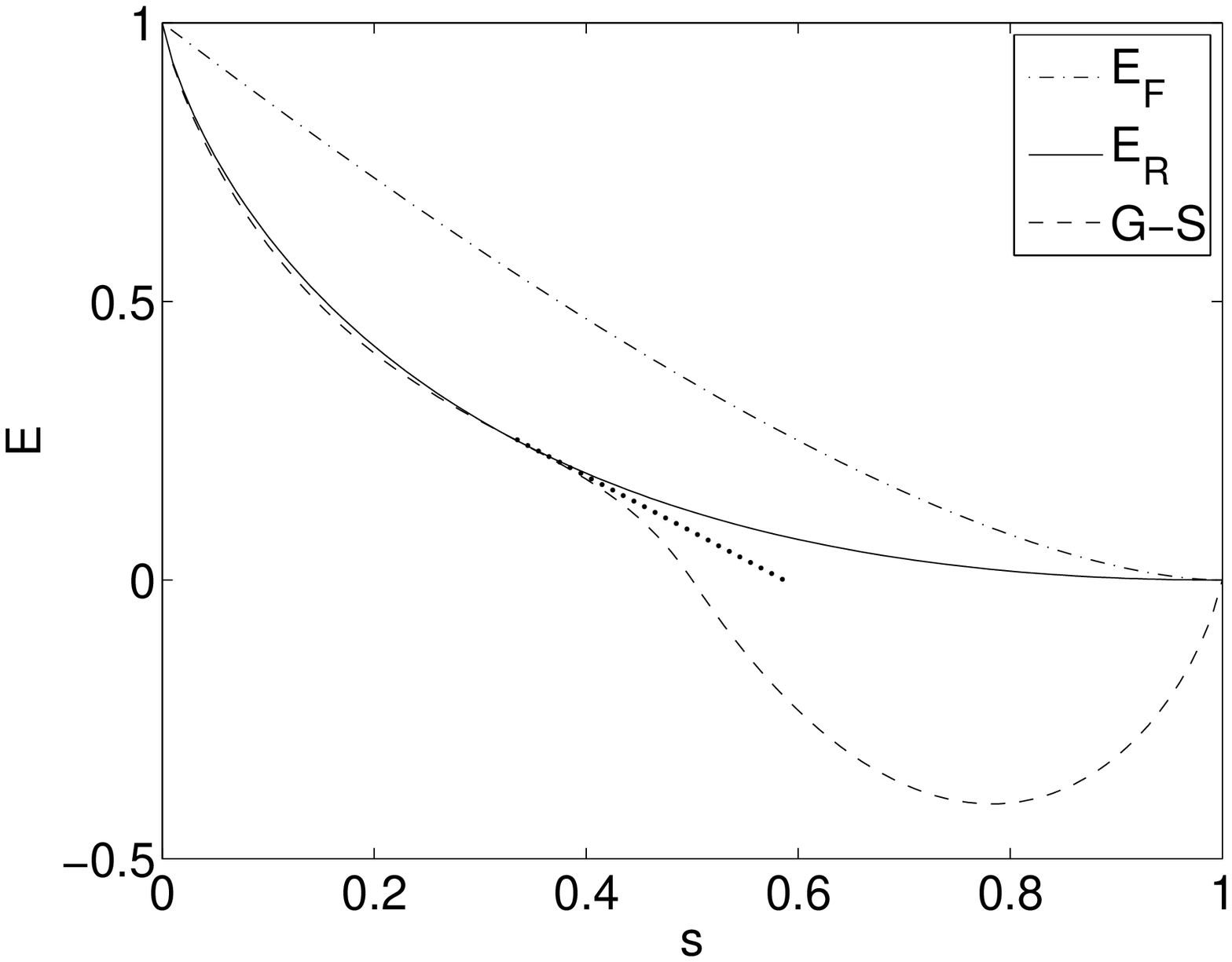}
   \includegraphics[width=5.2cm]{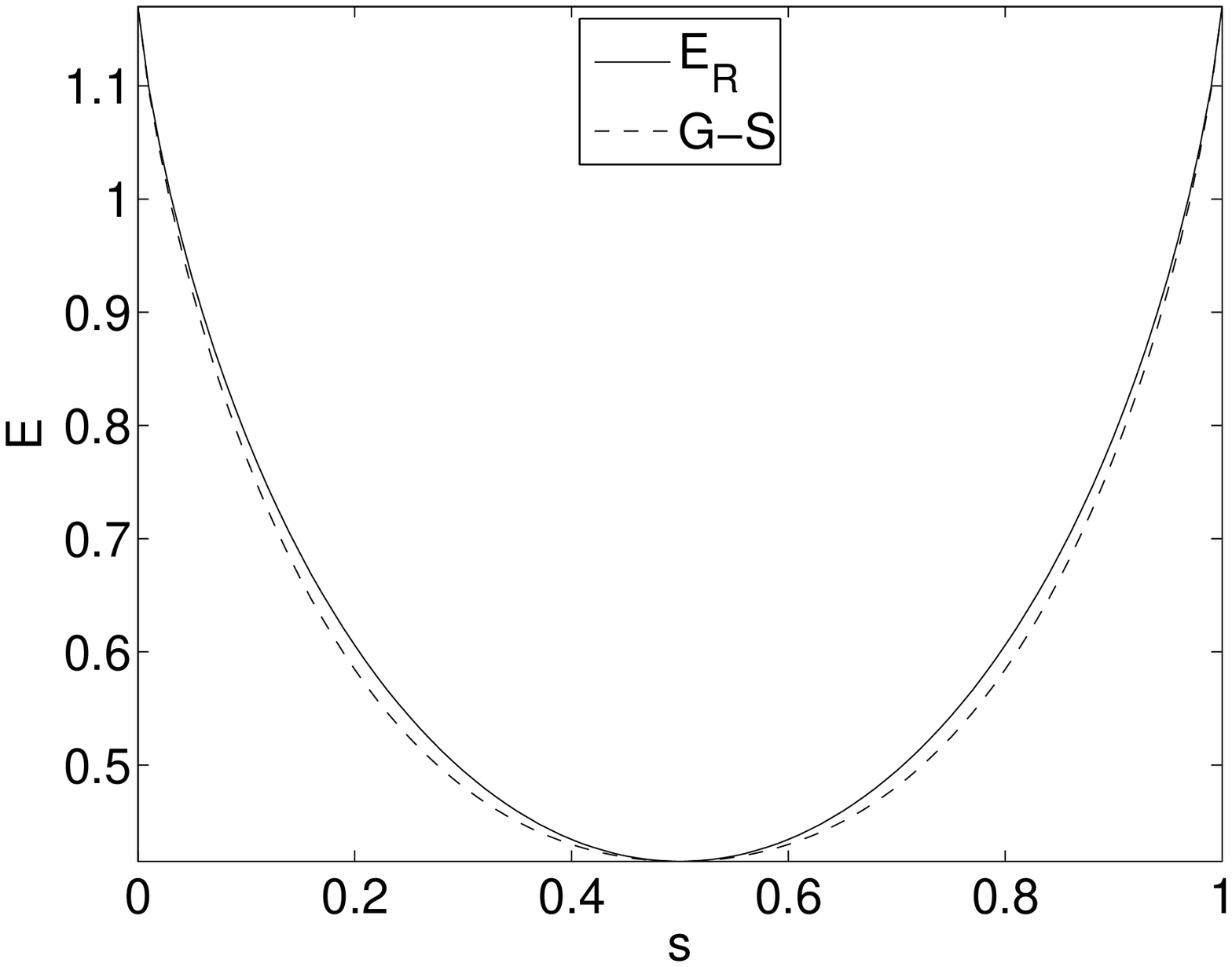}
    \includegraphics[width=5.2cm]{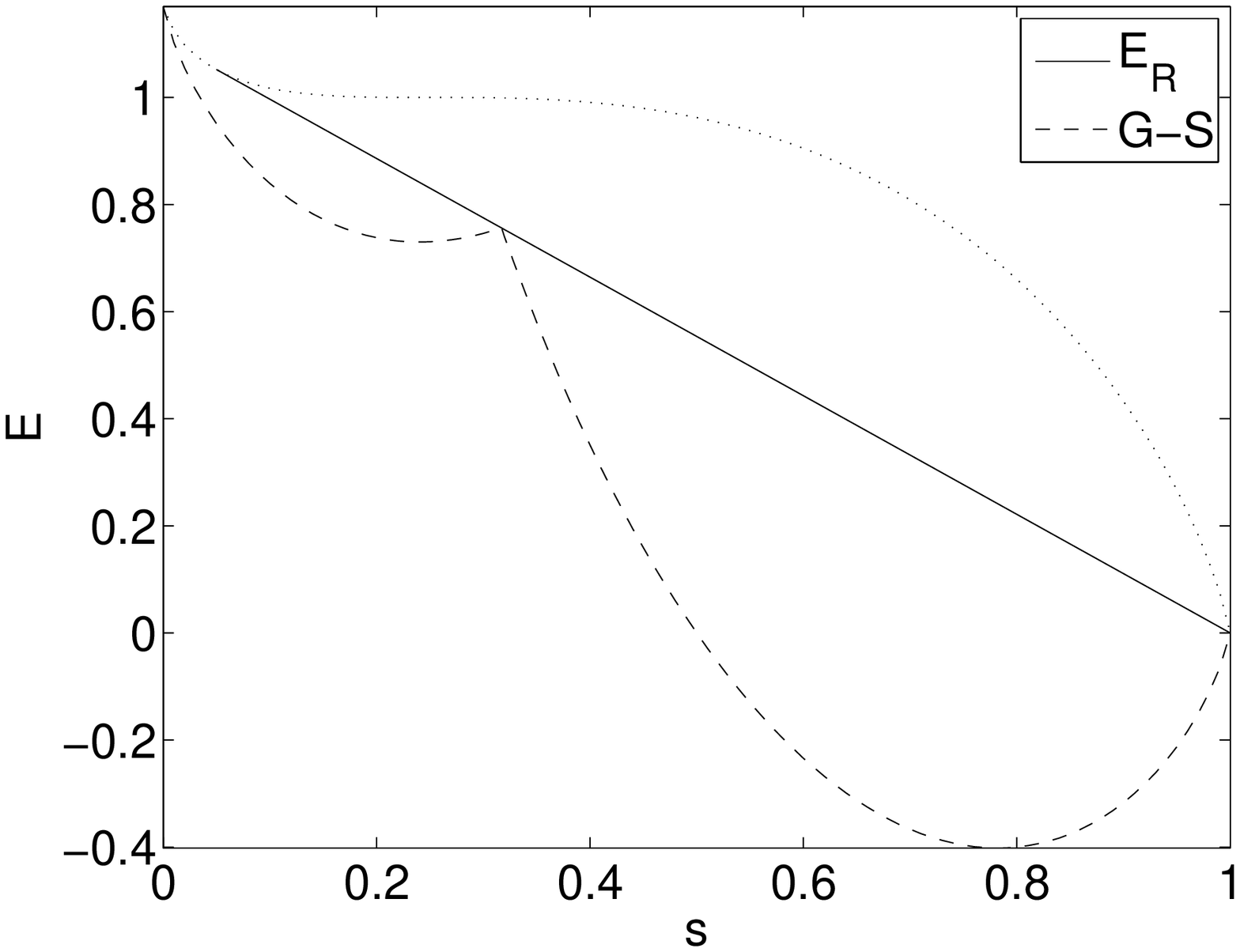}
  \caption{\label{fig:DickeMix}   REE and lower bound for AREE
  given by $G-S$ of three families of states,
  $\rho_{2;0,1}(s)$ (left plot),
$\rho_{3;1,2}(s)$ (middle plot), and $\rho_{3;0,2}(s)$ (right plot),
respectively, see \eref{eq:DickeMix} for the definitions of these
states. In the left plot, entanglement of formation $E_{\mathrm{F}}$
is also plotted for comparison; the dotted line is the improved
lower bound for AREE after taking the convexity into account. In the
right plot, REE is obtained by convex roof construction from the
dotted curve \cite{wei08}. After taking the convexity of AREE into
account, the lower bound for  AREE  derived from $G-S$ is almost
equal to REE.}
\end{figure}

\Fref{fig:DickeMix} illustrates $E_{\mathrm{R}}$ (REE is given by
theorem 1 of Wei \cite{wei08}) and $G-S$ for the following three
families of states:
\begin{eqnarray}
\label{eq:DickeMix}
&&\rho_{2;0,1}(s)=s|11\rangle\langle11|+(1-s)|\Psi_2\rangle\langle\Psi_2|,\nonumber\\
&&\rho_{3;1,2}(s)=s|\widetilde{\mathrm{W}}
\rangle\langle\widetilde{\mathrm{W}}|+(1-s)|\mathrm{W}\rangle\langle\mathrm{W}|,\nonumber\\
&&\rho_{3;0,2}(s)=s|111\rangle\langle111|+(1-s)|\mathrm{W}\rangle\langle\mathrm{W}|,
\end{eqnarray}
where $|\Psi_2\rangle=\frac{1}{\sqrt{2}}(|01\rangle+|10\rangle)$,
$\mathrm{W}=\frac{1}{\sqrt{3}}(|100\rangle+|010\rangle+|001\rangle)$,
and
$|\widetilde{\mathrm{W}}\rangle=\frac{1}{\sqrt{3}}(|011\rangle+|101\rangle+|110\rangle)$.
For the first family of states $\rho_{2;0,1}(s)$ (left plot of
\fref{fig:DickeMix}), $G-S$ gives a very good lower bound for AREE
when $0\leq s\leq 0.4$. The bound is tight at $s=\frac{1}{3}$, since
$\rho_{2;0,1}(\frac{1}{3})$ is the bipartite  reduced state of the
Dicke state $|\widetilde{\mathrm{W}}
\rangle\langle\widetilde{\mathrm{W}}|$. Taking the convexity of AREE
into account, we can raise the lower bound for $s>\frac{1}{3}$ to
the one represented by the dotted line, which is tangent to both the
curve $E_{\mathrm{R}}(s)$ and the curve $G(s)-S(s)$ at
$s=\frac{1}{3}$. In addition, $G-S$   is a lower bound for
entanglement cost. For $\rho_{3;1,2}(s)$ (middle plot), the bound is
very good in the whole parameter region. The bound is tight at
$s=\frac{1}{2}$, since $\rho_{3;1,2}(\frac{1}{2})$ is the tripartite
reduced state of the Dicke state $|4,2\rangle\langle4,2|$.  For
$\rho_{3;0,2}(s)$ (right plot), REE is obtained by convex roof
construction from the dotted curve as described in \cite{wei08}. The
lower bound for AREE given by $G(s)-S(s)$ does not look very good at
first glance. However, taking the convexity of AREE into account, we
can obtain a  lower bound for AREE which is very close  to REE for
almost entire family of states $\rho_{3;0,2}(s)$.

\subsection{\label{sec:add-Smo}  The Smolin state}

The Smolin state  is a four-qubit unlockable bound entangled state,
from which no pure entanglement can be distilled under LOCC.
However, if any two of the four parties come together, they can
create a singlet between the other two parties \cite{Smolin01}. The
Smolin state can be expressed in several equivalent forms, one of
which is
\begin{eqnarray}
\rho_{ABCD}=\frac{1}{4}\sum_{j=0}^3\bigl(|\Psi_j\rangle\langle\Psi_j|)_{AB}\otimes(|\Psi_j\rangle\langle\Psi_j|\bigr)_{CD},
\end{eqnarray}
where $|\Psi_j\rangle$s are the four Bell states
$\frac{1}{\sqrt{2}}(|00\rangle\pm|11\rangle)$ and
$\frac{1}{\sqrt{2}}(|01\rangle\pm|10\rangle)$. It can also be
written in a more symmetric form
\begin{eqnarray}
\label{eq:Smolin}
\rho_{ABCD}=\frac{1}{16}\Bigl(I^{\otimes4}+\sum_{j=1}^3\sigma_j^A\otimes\sigma_j^B\otimes\sigma_j^C\otimes\sigma_j^D\Bigr),
\end{eqnarray}
which clearly shows that it is permutation invariant and
non-negative.

Since its discovery, the Smolin state has found many applications,
such as remote information concentration \cite{mv01},
superactivation \cite{sst03}, and multiparty secret sharing
\cite{ah06}. It can maximally violate a two-setting Bell inequality
similar to the CHSH inequality \cite{ah06b}. It was also used to
show that four orthogonal Bell states cannot be discriminated
locally even probabilistically \cite{gkrss01}. Recently,  Amselem
and Bourennane have realized the Smolin state in experiments with
polarized photons and characterized its entanglement properties
\cite{ab09}. Similar experiments were performed later by several
other groups \cite{lkp10,bsg10}. Hence, it is desirable to quantify
the amount of entanglement in the Smolin state.

The multipartite REE of the Smolin state has been derived by Murao
and Vedral \cite{mv01} and by Wei {\it et al} \cite{wag04}, with the
result $E_{\mathrm{R}}(\rho_{ABCD})=1$. The  derivation in
\cite{wag04} relies on the following alternative representation of
the Smolin state, which again shows that it  is non-negative,
\begin{eqnarray}
\rho_{ABCD}=\frac{1}{4}\sum_{j=0}^3|X_j\rangle\langle X_j|,
\end{eqnarray}
with
\begin{eqnarray*}
|X_0\rangle=\frac{1}{\sqrt{2}}(|0000\rangle+|1111\rangle),\qquad
|X_1\rangle=\frac{1}{\sqrt{2}}(|0011\rangle+|1100\rangle),\\
|X_2\rangle=\frac{1}{\sqrt{2}}(|0101\rangle+|1010\rangle),\qquad
|X_3\rangle=\frac{1}{\sqrt{2}}(|0110\rangle+|1001\rangle).
\end{eqnarray*}
They also give a closest separable state to $\rho_{ABCD}$, which
reads
\begin{eqnarray}
\rho_{\mathrm{sep}}&=&\frac{1}{8}\bigl(|0000\rangle\langle0000|+|1111\rangle\langle1111|+|0011\rangle\langle0011|+|1100\rangle\langle1100|\nonumber\\
&&+|0101\rangle\langle0101|
{}+|1010\rangle\langle1010|+|0110\rangle\langle0110|+|1001\rangle\langle1001|\bigl).
\end{eqnarray}
Note that
$\rho_{\mathrm{sep}}=\frac{1}{2}(\rho_{ABCD}+\rho_{\perp})$, where
$\rho_{\perp}$ is orthogonal to $\rho_{ABCD}$, hence
$R_{\mathrm{L}}(\rho_{ABCD})\leq1$ according to \eref{eq:GR} and
\eref{eq:LGR}. Since $R_{\mathrm{L}}(\rho_{ABCD})\geq
E_{\mathrm{R}}(\rho_{ABCD})=1$, we get
$R_{\mathrm{L}}(\rho_{ABCD})=1$.

To compute GM of the Smolin state, note that the closest product
state to $\rho_{ABCD}$ can be chosen to be  non-negative, according
to lemma~\ref{le:closest2}. Suppose
$|\varphi_4\rangle=\bigotimes_{j=1}^4(c_j|0\rangle+s_j|1\rangle)$ is
a closest product state, where $c_j=\cos\theta_j$,
$s_j=\sin\theta_j$ with $0\leq\theta_j\leq\frac{\pi}{2}$ for
$j=1,2,3,4$.
\begin{eqnarray}
\Lambda^2(\rho_{ABCD})&=&\langle\varphi_4|\rho_{ABCD}|\varphi_4\rangle\nonumber\\
&=&\frac{1}{8}\Bigl[(c_1c_2c_3c_4+s_1s_2s_3s_4)^2+(c_1c_2s_3s_4+s_1s_2c_3c_4)^2\nonumber\\
&&{}+(c_1s_2c_3s_4+s_1c_2s_3c_4)^2+(c_1s_2s_3c_4+s_1c_2c_3s_4)^2\Bigr]\leq\frac{1}{8},
\end{eqnarray}
where the last inequality was derived in \cite{wag04}. The same
result can also be obtained with the approach presented in
~\cite{lkp10}. Since
$\Lambda^2(\rho_{ABCD})\geq\langle0000|\rho_{ABCD}|0000\rangle
=\frac{1}{8}$, we thus  obtain $\Lambda^2(\rho_{ABCD})=\frac{1}{8}$
and $G(\rho_{ABCD})=3$. Note that $S(\rho_{ABCD})=2$,
$R_{\mathrm{L}}(\rho_{ABCD})=E_{\mathrm{R}}(\rho_{ABCD})=G(\rho_{ABCD})-S(\rho_{ABCD})$,
and $\rho_{ABCD}$ is non-negative. According to
theorems~\ref{thm:additivityGnon-negative2} and \ref{thm:AREE}, we
have
\begin{proposition}
  The Smolin state has strong additive GM, additive REE and LGR, and
  thus
\begin{eqnarray}
  G^\infty(\rho_{ABCD})  &=& G(\rho_{ABCD})  =3,\nonumber\\
 E_{\mathrm{R}}^\infty(\rho_{ABCD})&=& E_{\mathrm{R}}(\rho_{ABCD})=1, \nonumber\\
  R_{\mathrm{L}}^\infty(\rho_{ABCD})&=& R_{\mathrm{L}}(\rho_{ABCD})=1.
\end{eqnarray}
\end{proposition}

The additivity of REE of  the  Smolin state can  also be derived in
an alternative way by first considering REE under the bipartite cut
$A:BCD$ \cite{wag04}. Since  every pure state in the support of
$\rho_{A:BCD}$ is maximally entangled,  the entanglement of
formation of the state is given by $E_{\mathrm{F}}(\rho_{A:BCD})=1$.
On the other hand $E_{\mathrm{D}}(\rho_{A:BCD})\geq1$, where
$E_{\mathrm{D}}$ denotes entanglement of distillation, because a
singlet can be distilled from the Smolin state when any two of the
four parties come together. From the chain of inequalities,
$E_{\mathrm{F}}(\rho_{A:BCD})\geq E_{\mathrm{c}}(\rho_{A:BCD})\geq
E_{\mathrm{R}}^\infty(\rho_{A:BCD})\geq
E_{\mathrm{D}}(\rho_{A:BCD})$, it follows that
$E_{\mathrm{c}}(\rho_{A:BCD})=E_{\mathrm{R}}^\infty(\rho_{A:BCD})=
E_{\mathrm{D}}(\rho_{A:BCD})=1$. The additivity of REE then follows
from the following chain of inequalities: $1\leq
E_{\mathrm{R}}^\infty(\rho_{A:BCD})\leq
E_{\mathrm{R}}^\infty(\rho_{ABCD})\leq
E_{\mathrm{R}}(\rho_{ABCD})=1$.

Recall that, under asymptotic non-entangling operations, state
transformation can be made reversible, and AREE determines  the
transformation rate \cite{bp08}. Hence, the Smolin state and the
four-qubit GHZ state can be  converted into each other reversibly
under these operations.

\subsection{\label{sec:add-Dur}  D\"ur's multipartite  entangled states}

D\"ur's multipartite bound entangled state $\rho_N$ was found in
search of the relation between distillability of multipartite
entangled states and violation of Bell's inequality \cite{dur01}.
\begin{eqnarray}
\label{eq:Dur}
\rho_N=\frac{1}{N+1}\biggl[|\Psi_G\rangle\langle\Psi_G|
+\frac{1}{2}\sum_{k=1}^N\bigl(P_k+\bar{P}_k\bigr)\biggr],
\end{eqnarray}
where $|\Psi_G\rangle=\frac{1}{\sqrt{2}}\bigl(|0^{\otimes
N}\rangle+\rme^{\rmi\alpha_N}|1^{\otimes N}\rangle\bigr)$ is the
$N$-partite GHZ state, $P_k$ is the projector onto the product state
$|u_k\rangle=|0\rangle_{A_1}|0\rangle_{A_2}\cdots
|1\rangle_{A_k}\cdots|0\rangle_{A_N}$, and $\bar{P}_k$ is the
projector onto the product state
$|v_k\rangle=|1\rangle_{A_1}|1\rangle_{A_2}\cdots
|0\rangle_{A_k}\cdots|1\rangle_{A_N}$. D\"ur has shown that, for
$N\geq4$, the state in \eref{eq:Dur} is bound entangled and, for
$N\geq8$, it violates two-setting Mermin-Klyshko-Bell inequality
\cite{dur01}.  Since the phase factor $\rme^{\rmi\alpha_N}$ can be
absorbed by redefining the computational basis, we may assume
$\rme^{\rmi\alpha_N}=1$ without loss of generality. It is then clear
that  $\rho_N$ is non-negative. In the following discussion, we
assume $N\geq4$.

Wei {\it et al} \cite{wag04} have generalized D\"ur's multipartite
bound entangled state to the following family of states:
\begin{eqnarray}\label{eq:Dur2}
\rho_N(x)=x|\Psi_G\rangle\langle\Psi_G|
+\frac{1-x}{2N}\sum_{k=1}^N\bigl(P_k+\bar{P}_k\bigr),
\end{eqnarray}
and shown that the state is bound entangled if $0\leq
x\leq\frac{1}{N+1}$ and free entangled if $\frac{1}{N+1}< x\leq1$.
Moreover, they had conjectured REE of this state to be
\begin{eqnarray}
\label{eq:DurREE}
 E_{\mathrm{R}}[\rho_N(x)]=x \quad\mbox{for} \quad N\geq4,
\end{eqnarray}
which was later proved in \cite{wei08}.

We shall show that REE of $\rho_N(x)$ is  additive by first showing
that REE of   $\rho_N=\rho_N(\frac{1}{N+1})$ is additive, and then
extending the result to the whole family of states via the convexity
of AREE. Note that $\rho_N(x)$ is a convex combination of
$\rho_N(0)$ and $\rho_N(1)$, that is,
$\rho_N(x)=x\rho_N(1)+(1-x)\rho_N(0)$.

Since  $\rho_N(x)$ is non-negative, its closest product state can be
chosen to be non-negative,  according to lemma~\ref{le:closest2}.
Let
$|\varphi_N\rangle=\bigotimes_{j=1}^N(c_j|0\rangle+s_j|1\rangle)$ be
a closest product state, where $c_j=\cos\theta_j$,
$s_j=\sin\theta_j$ with $0\leq\theta_j\leq\frac{\pi}{2}$ for
$j=1,2,\ldots,N$.
\begin{eqnarray}
\fl\Lambda^2[\rho_N(x)]&=& \max_{|\varphi_N\rangle}
\langle\varphi_N|\rho_N(x)|\varphi_N\rangle,\nonumber\\
\fl&=& \max_{\theta_1,\ldots,\theta_N}
\biggl\{\frac{x}{2}(c_1\cdots c_N+s_1\cdots s_N)^2\nonumber\\
\fl&&+\frac{1-x}{2N}\sum_{k=1}^N\bigl[(c_1\cdots c_{k-1}s_k
c_{k+1}\cdots c_N)^2+(s_1\cdots s_{k-1}c_k s_{k+1}\cdots
s_N)^2\bigr]\biggr\}.
\end{eqnarray}
When $x=\frac{1}{N+1}$, $\rho_N(x)=\rho_N$, according to Appendix A
in \cite{wag04}, $\Lambda^2(\rho_N)\leq\frac{1}{2(N+1)}$; since the
product state $|0^{\otimes N}\rangle$ achieves this bound, it
follows that $\Lambda^2(\rho_N)=\frac{1}{2(N+1)}$ and thus
$G(\rho_N)=\log2(N+1)$. In addition,  $|u_k\rangle, |v_k\rangle,
\forall k$ and $|0^{\otimes N}\rangle, |1^{\otimes N}\rangle $ are
 closest product states to $\rho_N$. When $0\leq x \leq
\frac{1}{N+1}$, $|u_k\rangle, |v_k\rangle, \forall k$ are still
closest product states to $\rho_N(x)$; by contrast, when
$\frac{1}{N+1}\leq x\leq1$, $|0^{\otimes N}\rangle, |1^{\otimes
N}\rangle $ are still closest product states to $\rho_N(x)$. So, for
$N\geq4$, we obtain
\begin{eqnarray}
\Lambda^2[\rho_N(x)]=\left\{\begin{array}{cc}
                            \frac{1-x}{2N}, &0\leq x\leq \frac{1}{N+1}, \\
                            \frac{x}{2},& \frac{1}{N+1}\leq x\leq 1.
                          \end{array}\right.
\end{eqnarray}

\begin{figure}
  \centering
  \includegraphics[width=8cm]{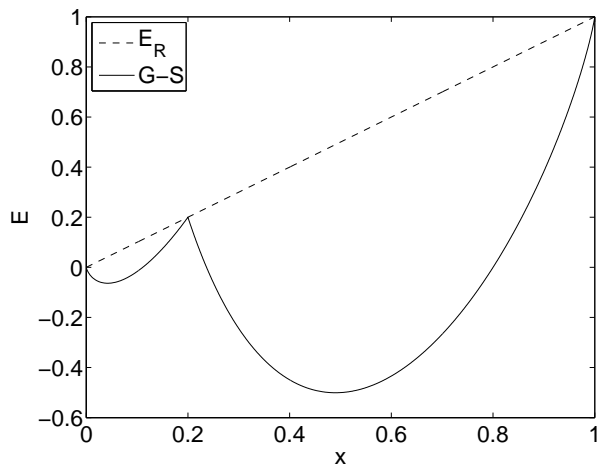}\\
    \caption{\label{fig:Dur}   REE  and lower bound for AREE given by
  $G-S$ of D\"ur's multipartite
  entangled state $\rho_N(x)$ with $N=4$.
 The lower bound is tight only at three points $x=0, \frac{1}{N+1}, 1$, but it is enough to infer the additivity of REE
due to the convexity of  AREE. }
\end{figure}

Meanwhile, according to theorems~\ref{thm:additivityGnon-negative2}
and \ref{thm:AREE}, GM of D\"ur's multipartite  entangled states is
strong additive, and $G[\rho_N(x)]-S[\rho_N(x)]$ gives a lower bound
for $E_{\mathrm{R}}^\infty[\rho_N(x)]$. The bound is tight at the
following three points $x=0, 1, \frac{1}{N+1}$, hence
$E_{\mathrm{R}}^\infty[\rho_N(0)\bigr]=E_{\mathrm{R}}\bigl[\rho_N(0)\bigr]=0$,
$E_{\mathrm{R}}^\infty\bigl[\rho_N(1)\bigr]=E_{\mathrm{R}}\bigl[\rho_N(1)\bigr]=1$,
$E_{\mathrm{R}}^\infty(\rho_N)=E_{\mathrm{R}}(\rho_N)=\frac{1}{N+1}$,
(see also \fref{fig:Dur}). Although the bound is in general not
tight, from the convexity of AREE, we can already conclude that
$E_{\mathrm{R}}^\infty[\rho_N(x)]=x$.
\begin{proposition}
  GM  and  REE of  D\"ur's multipartite entangled state $\rho_N(x)$ with $N\geq4$ are strong additive  and additive, respectively,
  and thus
\begin{eqnarray}
G^\infty[\rho_N(x)]     &=& G[\rho_N(x)]=\left\{\begin{array}{cc}
                            -\log\frac{1-x}{2N}, &0\leq x\leq \frac{1}{N+1}, \\
                           -\log \frac{x}{2},& \frac{1}{N+1}\leq x\leq
                           1.
                          \end{array}\right.\nonumber\\
  E_{\mathrm{R}}^\infty[\rho_N(x)] &=& E_{\mathrm{R}}[\rho_N(x)]=x.
\end{eqnarray}
\end{proposition}

\section{\label{sec:NonAdditivity}
Non-additivity of geometric measure  of antisymmetric states}

In this section, we turn to the antisymmetric subspace, and explore
the  connection between the permutation symmetry and the additivity
property of multipartite entanglement measures. Starting from a
simple observation on the closest product states to antisymmetric
states and that to symmetric states, we show that GM is non-additive
for all antisymmetric states shared over three or more parties, and
provide a unified explanation of the non-additivity of the three
measures GM, REE and LGR of the antisymmetric projector states. In
particular, we establish a simple equality among the three measures
GM, REE and LGR of the tensor product of antisymmetric projector
states, and derive analytical formulae of the three measures in the
case of one copy and two copies, respectively.  Our results may be
found useful in the study of fermion systems, which are described by
antisymmetric states due to the super-selection rule.

In \sref{sec:ass1}, we introduce Slater determinant states, which
are analog of product states in the antisymmetric subspace, and give
a simple criterion on when an antisymmetric state is a Slater
determinant state. Then we prove that the $N$ one-particle reduced
states of each closest product state to any $N$-partite
antisymmetric state are mutually orthogonal, and derive a lower
bound for the three measures GM, REE and LGR  based on this
observation. In \sref{sec:ass2}, we show that GM of antisymmetric
states shared over three or more parties is non-additive. In
\sref{sec:ass3}, we establish a simple equality among the three
measures GM, REE and LGR of the tensor product of antisymmetric
projector states, and compute the three measures in the case of one
copy and two copies respectively.   REE and LGR of the mixture of
Slater determinant states are also derived. In \sref{sec:ass4}, we
treat generalized antisymmetric states \cite{bra03} as further
counterexamples to the additivity of GM.

\subsection{\label{sec:ass1}  Geometric measure of antisymmetric states}

We shall be concerned with antisymmetric states in the multipartite
Hilbert space  $\mathcal{H}=\bigotimes^N_{j=1} \mathcal{H}_j$ with
$\mathrm{Dim} {\mathcal{H}_j}=d$ and $N\leq d$. A ket
$|\psi_N\rangle$ is \emph{antisymmetric} if every odd permutation of
the  parties induces a sign change. All unnormalized antisymmetric
kets form the antisymmetric subspace $\mathcal{H}_-$, whose
dimension is ${d\choose N}=d!/[N!(d-N)!]$. An $N$-partite state
$\rho_N$ is antisymmetric if its support is contained in the
antisymmetric subspace. Let $P_{d,N}$ be the projector onto the
antisymmetric subspace $\mathcal{H}_-$; then
$\mathrm{tr}(P_{d,N})=\mathrm{Dim} \mathcal{H}_-={d\choose N}$. Any
$N$-partite state $\rho_N$ is an antisymmetric state if and only if
the equality $\rho_N=P_{d,N}\rho_N P_{d,N}$ holds. A typical example
of antisymmetric states is the \emph{antisymmetric projector state}
$\rho_{d,N}=\frac{P_{d,N}}{\mathrm{tr}(P_{d,N})}$, which includes
the antisymmetric basis state and antisymmetric Werner state as
special cases.

Given  $N$  orthonormal single-particle states, $|a_1\rangle,
\ldots, |a_N\rangle$, a \emph{Slater determinant state} can be
constructed by anti-symmetrization, a procedure routinely used in
the study of fermion systems, i.e.
\begin{eqnarray}
\label{eq:antisymmetric} &&|a_1\rangle\wedge \cdots
\wedge|a_N\rangle := \frac{1}{\sqrt{N!}}\sum_{\sigma\in
S_N}\mathrm{sgn}(\sigma)|a_{\sigma(1)},a_{\sigma(2)},\ldots,a_{\sigma(N)}\rangle,
\end{eqnarray}
where $S_N$ is the symmetry group of  $N$ letters,
$\mathrm{sgn}(\sigma)$ is the signature of $\sigma$ \cite{hj85}, and
 $\frac{1}{\sqrt{N!}}$ is the  normalization factor.
Apparently, all Slater determinant states are locally unitarily
equivalent to each other. In particular, they are locally unitarily
equivalent to  \emph{antisymmetric basis states},
$|j_1\rangle\wedge\cdots\wedge|j_N\rangle$ with $0\leq j_1<\cdots<
j_N\leq d-1$, which form an orthonormal basis in the antisymmetric
subspace. When $d=N$, there is only one antisymmetric basis state,
\begin{eqnarray}\label{eq:abs}
  &&\ket{\psi_{N-}}:=|0\rangle\wedge|1\rangle\wedge\cdots\wedge
  |N-1\rangle.
\end{eqnarray}

For the convenience of the following discussion, we summarize a few
useful properties of Slater determinant states; see \cite{Bha97} for
some mathematical background. If the $N$ single-particle states
$|a_1\rangle, \ldots, |a_N\rangle$ are linearly dependent, then
$|a_1\rangle\wedge \cdots \wedge|a_N\rangle$ vanishes. If they are
linearly independent but not mutually orthogonal, $|a_1\rangle\wedge
\cdots \wedge|a_N\rangle$ is a subnormalized Slater determinant
state. In that case, we can choose $N$ orthonormal states
$|a_1^\prime\rangle,\ldots, |a_N^\prime\rangle$ from the span of
$|a_1\rangle, \ldots, |a_N\rangle$, such that $|a_1\rangle\wedge
\cdots \wedge|a_N\rangle=c |a_1^\prime\rangle\wedge \cdots
\wedge|a_N^\prime\rangle$, where $c$ is a constant with modulus
between 0 and 1. The projection of a generic pure product state onto
the antisymmetric subspace is a subnormalized Slater determinant
state, that is,
\begin{eqnarray}
P_{d,N}(|a_1\rangle\otimes \cdots
\otimes|a_N\rangle)=\frac{1}{\sqrt{N!}}|a_1\rangle\wedge \cdots
\wedge|a_N\rangle.
\end{eqnarray}
Suppose  $|b_1\rangle,\ldots, |b_N\rangle$ are another  $N$
normalized single-particle states. Then $|a_1\rangle\wedge \cdots
\wedge|a_N\rangle$ and $|b_1\rangle\wedge \cdots \wedge|b_N\rangle$
are linearly independent if and only if the subspaces spanned by
$|a_1\rangle, \ldots, |a_N\rangle$ and by $|b_1\rangle, \ldots,
|b_N\rangle$, respectively, are different but of the same dimension
$N$. In other words, up to  overall phase factors, there is a
one-to-one correspondence between $N$-partite Slater determinant
states and $N$-dimensional  subspaces of the single-particle Hilbert
space.

Slater determinant states play a similar role in the antisymmetric
subspace as product states do in the full Hilbert space
\cite{sck01}. Given an $N$-partite antisymmetric state
$|\psi_N\rangle$, a basic task is to determine whether it is a
Slater determinant state. Note that the one-particle reduced state
of any $N$-partite Slater determinant state is a subnormalized
projector with rank $N$. On the other hand, if the one-particle
reduced state of an  antisymmetric state is of rank $N$, then there
is only one linearly independent Slater determinant state that  can
be constructed from the one-particle states in the support of this
one-particle reduced state. Obviously, the rank of the one-particle
reduced state can not be less than $N$; otherwise, no Slater
determinant state can be constructed. So we obtain
\begin{proposition}\label{pro:Slater} The one-particle reduced state of any $N$-partite antisymmetric
state has rank at least $N$. Moreover, an antisymmetric state is a
Slater determinant state if and only if its one-particle reduced
state has rank $N$.
\end{proposition}

We are now ready to study GM of antisymmetric states. Suppose
$\rho_N$ is an $N$-partite antisymmetric state and thus
$P_{d,N}\rho_NP_{d,N}=\rho_N$. Let
 $\varphi_N=|a_1\rangle\otimes\cdots\otimes|a_N\rangle$,
\begin{eqnarray}\label{eq:ass1}
\Lambda^2(\rho_N)&=&\max_{|\varphi_N\rangle}
\langle\varphi_N|\rho_N|\varphi_N\rangle= \max_{|\varphi_N\rangle}
\langle\varphi_N|P_{d,N}\rho_NP_{d,N}|\varphi_N\rangle,\nonumber\\
&=&\frac{1}{N!}\max_{|a_1\rangle,\ldots,|a_N\rangle}
\quad\bigl(\langle a_1|\wedge \cdots\wedge\langle
a_N|\bigr)\rho_N\bigl(|a_1\rangle\wedge \cdots
\wedge|a_N\rangle\bigr).
\end{eqnarray}
Recall that $|a_1\rangle\wedge \cdots \wedge|a_N\rangle$ is in
general a subnormalized Slater determinant state, and that it is
normalized if and only if the $N$ single-particle states
$|a_1\rangle,\ldots,|a_N\rangle$ are orthonormal, which is also a
necessary condition for $|\varphi_N\rangle$ to be a closest product
state.
\begin{proposition}\label{pro:ass}
The $N$ one-particle reduced states of any closest product state to
an $N$-partite antisymmetric state are mutually orthogonal.
\end{proposition}
Thus   searching for the closest product state of $\rho_N$ is
equivalent to searching for its closest Slater determinant state. A
peculiar feature of an antisymmetric state $\rho_N$ is the high
degeneracy of its closest product states.  If $|a_1\rangle\otimes
\cdots \otimes|a_N\rangle$ is a closest product state, then the
tensor product of any $N$ orthonormal states from the span of the
$N$ single-particle states $|a_1\rangle, \ldots, |a_N\rangle$ is
also a closest product state. Recall that there is a one-to-one
correspondence between $N$-partite Slater determinant states and
$N$-dimensional subspaces of the single-particle Hilbert space.

Proposition~\ref{pro:ass} is in a sense the analog of
proposition~\ref{pro:symmetric} for antisymmetric states. It is
crucial to computing GM of antisymmetric states and to proving the
non-additivity of GM of antisymmetric states shared over three or
more  parties in \sref{sec:ass2}.

Suppose $\lambda_{\mathrm{max}}$ is the largest eigenvalue of
$\rho_N$,  then $\Lambda^2(\rho_N)\leq
\frac{\lambda_{\mathrm{max}}}{N!}$ according to \eref{eq:ass1}, and
the inequality is saturated  if and only if there is  a Slater
determinant state in the eigenspace corresponding to
$\lambda_{\mathrm{max}}$. So we obtain
\begin{eqnarray}\label{eq:lowerbound}
&&R_{\mathrm{L}}(\rho_N)\geq E_{\mathrm{R}}(\rho_N)\geq
G(\rho_N)-S(\rho_N)\geq-\log\frac{\lambda_{\mathrm{max}}}{N!}-S(\rho_N).
\end{eqnarray}
A typical example where all the inequalities are saturated is the
antisymmetric projector state, as we shall see in \sref{sec:ass3}.

\subsection{\label{sec:ass2}  Non-additivity theorem for geometric measure of antisymmetric states}
The permutation symmetry of multipartite states plays a crucial role
in determining the properties of their closest product states, as
demonstrated in propositions~\ref{pro:symmetric} and \ref{pro:ass}.
It is also closely related to the non-additivity of GM of
antisymmetric states \footnote{In his private communication to MH in
May 2009, Shashank Virmani  pointed out (as a consequence of the
result of H\"ubener {\it et al} \cite{hkw09}) the non-additivity of
GM of all pure antisymmetric states of three or more parties.}.
\begin{theorem}\label{thm:non-additivity}
When $N\geq3$, GM  is non-additive for any two $N$-partite
antisymmetric states $\rho_N$ and $\rho_N^\prime$, that is,
$G(\rho_N\otimes\rho_N^\prime)<G(\rho_N)+G(\rho_N^\prime)$.
\end{theorem}
\begin{proof}
Suppose there exists a closest product state of
$\rho_N\otimes\rho_N^\prime$ which is of the tensor-product form
$|\varphi_N\rangle\otimes|\varphi_N^\prime\rangle$, then
$|\varphi_N\rangle$ and $|\varphi_N^\prime\rangle$ are  closest
product states of $\rho_N$ and $\rho_N^\prime$, respectively. Since
the set of one-particle reduced states of $|\varphi_N\rangle$
($|\varphi_N^\prime\rangle$) are mutually orthogonal according to
proposition~\ref{pro:ass},
$|\varphi_N\rangle\otimes|\varphi_N^\prime\rangle$ cannot be
symmetric. On the other hand, $\rho_N \ox \r^\prime_N$ is a
symmetric state and, if $N\geq3$, its closest product states are
necessarily symmetric according to proposition~\ref{pro:symmetric},
hence a contradiction would arise. In other words,  no closest
product state of $\rho_N\otimes\rho_N^\prime$ can be written as a
tensor product of the closest product states of $\rho_N$ and
$\rho_N^\prime$, respectively, which implies that
$G(\rho_N\otimes\rho_N^\prime)<G(\rho_N)+G(\rho_N^\prime)$.
\end{proof}

The  non-additivity of GM of antisymmetric states can be understood
as follows.  Antisymmetric states are generally more entangled than
symmetric states as noticed in \cite{hmm08}. However, two copies of
antisymmetric states turn to be a symmetric state.
Theorem~\ref{thm:non-additivity} establishes a simple connection
between permutation symmetry and the additivity property of GM of
multipartite states. In some special cases, this connection carries
over to other multipartite entanglement measures, such as REE and
LGR, as we shall see in \sref{sec:ass3}.

For a pure tripartite antisymmetric state, the non-additivity of GM
translates immediately to the non-multiplicativity of the maximum
output purity $\nu_\infty$ of the corresponding quantum channel
constructed according to the Werner-Holevo recipe. For example, the
non-multiplicativity  of the maximum output purity of the
Werner-Holevo channel is equivalent to the non-additivity of GM of
the tripartite antisymmetric basis state
 \cite{wh02}.

Theorem~\ref{thm:non-additivity} can be generalized to cover the
situation where the two states are not fully antisymmetric.
\begin{corollary}
GM is non-additive for two $N$-partite states, if there exists a
subsystem of three parties such that the respective tripartite
reduced states of the two $N$-partite states are both antisymmetric.
\end{corollary}
\begin{proof}
Assume $N>3$, suppose $\sigma_N$ and $\sigma_N^\prime$ are  two
$N$-partite states whose respective tripartite reduced states
$\sigma_N^{A_1,A_2,A_3}$ and ${\sigma_N^\prime}^{A_1,A_2,A_3}$ are
antisymmetric. Let $|a_1\rangle\otimes\cdots \otimes|a_N\rangle$ and
$|a_1^\prime\rangle\otimes\cdots \otimes|a_N^\prime\rangle$ be the
closest product states to $\sigma_N$ and $\sigma_N^\prime$,
respectively; then $(\langle a_4|\otimes\cdots\otimes\langle
a_N|)\sigma_N(|a_4\rangle\otimes \cdots\otimes |a_N\rangle)$ and
$(\langle a_4^\prime|\otimes\cdots\otimes\langle
a_N^\prime|)\sigma_N^\prime(|a_4^\prime\rangle\otimes \cdots\otimes
|a_N^\prime\rangle)$ are both antisymmetric.
Theorem~\ref{thm:non-additivity} applied to the two subnormalized
antisymmetric states shows that
$G(\sigma_N\otimes\sigma_N^\prime)<G(\sigma_N)+G(\sigma_N^\prime)$.
\end{proof}

In the bipartite scenario, if either $\rho_2$ or  $\rho_2^\prime$ is
pure, then
$G(\rho_2\otimes\rho_2^\prime)=G(\rho_2)+G(\rho_2^\prime)$, since GM
of bipartite pure states is strong additive, as shown in
\sref{sec:add-Bip}. On the other hand,  the closest product state to
$\rho_2\otimes\rho_2^\prime$ cannot be  of tensor-product form if it
is symmetric and vice versa, according to the same reasoning as that
in the proof of theorem~\ref{thm:non-additivity}. The additivity of
GM of $\rho_2$ and that of $\rho_2^\prime$ is  related to the
existence of  closest product states of $\rho_2\otimes\rho_2^\prime$
which are not symmetric. This in turn is due to the degeneracy of
Schmidt coefficients of $\rho_2$ or $\rho_2^\prime$ \cite{hj85}.
Indeed, every Schmidt coefficient of a bipartite  pure antisymmetric
state is at least doubly degenerate \cite{sck01}.

For generic bipartite antisymmetric states, we suspect that the
non-additivity of GM  is a rule rather than an exception, which is
supported by the following observation. If  both $\rho_2$ and
$\rho_2^\prime$ admit purifications that are antisymmetric, then
 their GM is non-additive,  due to
theorem~\ref{thm:non-additivity} and \eref{eq:ReduceGM}.

Theorem~\ref{thm:non-additivity} can also be derived in a slightly
different way, which offers a new perspective. According to
corollary 5 in \cite{hkw09} (see also
proposition~\ref{pro:symmetric} of this paper), the closest product
state to $\rho_N\otimes\rho_N^\prime$ can be chosen to be symmetric.
Let
\begin{eqnarray}\label{eq:V}
&& \ket{\varphi_N}=\ket{a_V}^{\otimes N},\qquad
\ket{a_V}=\frac{1}{\sqrt{d}}\sum_{j,k=0}^{d-1} V_{jk}|jk\rangle,
 \nonumber\\
&&V=\sum_{j,k=0}^{d-1}V_{jk}|j\rangle\langle k|
\qquad\mbox{with}\quad \mathrm{tr}(VV^\dag)=d.
\end{eqnarray}
According to \eref{eq:App:inner3} in the Appendix,
\begin{eqnarray}
\fl\Lambda^2(\rho_N\otimes\rho_N^\prime) &=&\max_{|\varphi_N\rangle}
\langle\varphi_N|\rho_{N}\otimes\rho_N^\prime|\varphi_N\rangle
=\max_{ \mathrm{tr}VV^\dag=d}
\,\frac{1}{d^N}\mathrm{tr}\bigl(\rho_{N}^{1/2}V^{\otimes
N}\rho_N^{\prime*} V^{\dag\otimes N}\rho_{N}^{1/2}\bigl)\label{eq:inner3}\\
\fl &=&\max_{\mathrm{tr}VV^\dag=d}
\,\frac{1}{d^N}\mathrm{tr}\bigl(\rho_{N}^{1/2}V^{\wedge
N}\rho_N^{\prime*} V^{\dag\wedge
N}\rho_{N}^{1/2}\bigl),\label{eq:inner4}
\end{eqnarray}
where $V^{\wedge N}=P_{d,N}V^{\otimes N}P_{d,N}$ is the restriction
of $V^{\otimes N}$ onto the antisymmetric subspace, which does not
vanish if and only if the rank of $V$ is at least $N$. Since the
rank of $V$ is exactly the  Schmidt rank of $|a_V\rangle$, the
Schmidt rank of $|a_V\rangle$  must be at least $N$, if
$|a_V\rangle^{\otimes N}$ is a closest product state. Recall that
the closest product state to $\rho_N\otimes\rho_N^\prime$ is
necessarily symmetric if $N\geq3$,  according to
proposition~\ref{pro:symmetric}. It follows that each closest
product state to $\rho_N\otimes\rho_N^\prime$ must be entangled
across the cut $A_1^1,\ldots,A_N^1:A_1^2,\ldots,A_N^2$, which
implies that
$G(\rho_N\otimes\rho_N^\prime)<G(\rho_N)+G(\rho_N^\prime)$.

In addition to providing an alternative proof of
theorem~\ref{thm:non-additivity}, the second approach also enables
us to compute GM of the antisymmetric projector states in
\sref{sec:ass3}, and to derive  an upper bound for  GM of
multipartite states of tensor-product form in \sref{sec:NonAddThm}.

\subsection{\label{sec:ass3}  Antisymmetric projector states}
In this section, we focus on the antisymmetric projector states,
which are typical examples of antisymmetric states, and include
antisymmetric basis states and antisymmetric Werner states as
special cases.  In particular, we establish a simple equality among
the three measures GM, REE and LGR of the tensor product of
antisymmetric projector states, and compute the three measures in
the case of one copy and two copies, respectively. Our study
provides a unified explanation of the non-additivity of the three
measures of the antisymmetric projector states.

The antisymmetric projector $P_{d,N}$ is invariant under the  action
of the unitary group $U(d)$ with the representation $U \mapsto
U^{\otimes N}$ for $U\in U(d)$. The range of $P_{d,N}$ is an
irreducible representation with multiplicity one \cite{hmm08}. In
other words, it satisfies the conditions (1) and  (5) of
proposition~\ref{ProH}. Moreover, the tensor product of the
antisymmetric projector states $\bigotimes_{j=1}^n\rho_{d_j,N}$
satisfies the conditions of proposition~\ref{ProH3}. So we obtain
\begin{proposition}\label{pro:ASPtensor}
GM, REE and LGR of antisymmetric projector states satisfy the
following equalities:
\begin{eqnarray}
&&R_{\mathrm{L}}\Biggl(\bigotimes_{j=1}^n\rho_{d_j,N}\Biggr)=E_{\mathrm{R}}\Biggl(\bigotimes_{j=1}^n\rho_{d_j,N}\Biggr)
=G\Biggl(\bigotimes_{j=1}^n\rho_{d_j,N}\Biggr)-\sum_{j=1}^{n}\log\mathrm{tr}P_{d_j,N},\nonumber\\
&&R_{\mathrm{L}}^\infty\bigl(\rho_{d,N}\bigr)=E_{\mathrm{R}}^\infty\bigl(\rho_{d,N}\bigr)=G^\infty\bigl(\rho_{d,N}\bigr)-\log\mathrm{tr}P_{d,N}.
\label{eq:equal}
\end{eqnarray}
\end{proposition}
Combining the above result with that on symmetric basis states
presented in \sref{sec:add-Dic}, we obtain
\begin{proposition}
The three measures GM, REE and LGR are equal  for the tensor product
of any number of symmetric basis states and antisymmetric basis
states, so are AGM, AREE and ALGR.
\end{proposition}

For the single copy antisymmetric projector state $\rho_{d,N}$, all
eigenvalues  are equal to $1/\mathrm{tr}(P_{d,N})$, and the
eigenspace corresponding to the largest eigenvalue of $\rho_{d,N}$
is exactly the antisymmetric subspace. Hence, all the inequalities
in \eref{eq:lowerbound} are saturated, which implies that
\begin{eqnarray}\label{eq:ASP1}
&&R_{\mathrm{L}}\bigl(\rho_{d,N}\bigr)=E_{\mathrm{R}}\bigl(\rho_{d,N}\bigr)=G\bigl(\rho_{d,N}\bigr)-\log\bigl[\mathrm{tr}\bigl(P_{d,N}\bigl)\bigr]=\log(N!).
\end{eqnarray}
Interestingly, REE and LGR of the antisymmetric projector state
$\rho_{d,N}$ do not depend on the dimension of the single-particle
Hilbert space. When $d=N$, the antisymmetric projector state turns
to be an antisymmetric basis state. The result on GM reduces to that
found in \cite{bra03, hmm08}, and  the result on REE and LGR reduces
to that found in \cite{weg04,wei08,hmm08}. When $N=2$, there is a
lower bound for AREE of $\rho_{d,N}$ found by Christandl {\it et al}
\cite{csw09} which reads
$E_{\mathrm{R}}^\infty(\rho_{d,2})\geq\log\sqrt{\frac{4}{3}}$, from
which we can get a lower bound for ALGR,
$R_{\mathrm{L}}^\infty\bigl(\rho_{d,2}\bigr)\geq\log\sqrt{\frac{4}{3}}$.
In general, none of the three measures is easy to compute for the
tensor product of antisymmetric projector states.

We  now focus on two copies of antisymmetric projector states. Note
that all entries of $P_{d,N}$ in the computational basis are real.
Let $|\varphi_N\rangle$ be as defined in \eref{eq:V}, according to
\eref{eq:App:inner3} in the Appendix,
\begin{eqnarray}
\label{eq:inner}
\fl\langle\varphi_N|P_{d,N}^{\otimes2}|\varphi_N\rangle=\frac{1}{d^N}\mathrm{tr}\bigl(P_{d,N}V^{\otimes
N}P_{d,N}V^{\dag\otimes
N}\bigl)=\frac{1}{d^N}\mathrm{tr}\bigl[P_{d,N}\bigl(VV^\dag\bigl)^{\otimes
N}P_{d,N}\bigl],
\end{eqnarray}
where in deriving the last equality, we have used the fact that
$V^{\otimes N}$ and $P_{d,N}$ commutes due to the Weyl reciprocity,
(see also \cite{Bha97}). The trace in \eref{eq:inner} is exactly the
$N$th symmetric polynomial of the set of eigenvalues $\mu_0,
\ldots,\mu_{d-1}$  of $VV^\dag$ \cite{Bha97} ($\mu_j\geq 0$,
$\sum_{j=0}^{d-1}\mu_j=\mathrm{tr}(VV^\dag)=d$), that is
\begin{eqnarray}
\label{eq:esp}
\langle\varphi_N|P_{d,N}^{\otimes2}|\varphi_N\rangle=\frac{1}{d^N}\sum_{0\leq
j_1<\cdots<j_N\leq d-1}\mu_{j_1}\cdots\mu_{j_N}.
\end{eqnarray}
Recall that elementary symmetric polynomials are Schur concave
functions \cite{bz06}, so the maximum in \eref{eq:esp} is obtained
if and only if $\mu_0=\mu_1=\cdots=\mu_{d-1}=1$, that is, $V$ is
unitary, or equivalently, $|a_V\rangle$ is  maximally entangled. So
we obtain
\begin{eqnarray}
&&\max_{|\varphi_N\rangle} \langle
\varphi_N|P_{d,N}^{\otimes2}|\varphi_N\rangle=\frac{1}{d^N}{d\choose
N}=\frac{d!}{d^NN!(d-N)!}.
\end{eqnarray}
In conjunction  with  \eref{eq:equal}, \eref{eq:ASP1} and
proposition~\ref{pro:symmetric} (see also corollary 5 in
\cite{hkw09}), we get
\begin{proposition}\label{pro:ASP2}
GM, REE and LGR of one copy and two copies of the antisymmetric
projector states are respectively given by
\begin{eqnarray}
\fl
R_{\mathrm{L}}\bigl(\rho_{d,N}\bigr)=E_{\mathrm{R}}\bigl(\rho_{d,N}\bigr)=G\bigl(\rho_{d,N}\bigr)-\log\bigl[\mathrm{tr}\bigl(P_{d,N}\bigl)\bigr]=\log(N!).\nonumber\\
\fl
R_{\mathrm{L}}\bigl(\rho_{d,N}^{\otimes2}\bigr)=E_{\mathrm{R}}\bigl(\rho_{d,N}^{\otimes2}\bigr)
=G\bigl(\rho_{d,N}^{\otimes2}\bigr)-2\log\bigl[\mathrm{tr}\bigl(P_{d,N}\bigl)\bigr]=\log\frac{d^NN!(d-N)!}{d!}.\label{eq:ASP2}
\end{eqnarray}
For $\rho_{d,N}$, a state is a closest product state if and only if
it is a tensor product of orthonormal single-particle states. For
$\rho_{d,N}^{\otimes2}$, any tensor product of identical maximally
entangled states across the cut $A_j^1:A_j^2$ for $j=1,\ldots,N$,
respectively, is a closest product state, and each closest product
state must be of this form if $N\geq3$.
\end{proposition}
GM, REE and LGR of $\rho_{d,N}$ are all non-additive if $d\geq3$ and
$2\leq N\leq d$. Moreover, REE (LGR) of $\rho_{d,N}^{\otimes2}$ is
almost equal to REE (LGR) of $\rho_{d,N}$ if $d\gg1$.

When $d=N$, $\rho_{d,N}=P_{d,N}=|\psi_{N-}\rangle\langle\psi_{N-}|$,
equation \eref{eq:ASP2} reduces to
\begin{eqnarray}\label{eq:asbs}
&&R_{\mathrm{L}}(|\psi_{N-}\rangle^{\otimes2})=E_{\mathrm{R}}(|\psi_{N-}\rangle^{\otimes2})=G(|\psi_{N-}\rangle^{\otimes2})=N\log
N.
\end{eqnarray}
Compared with \eref{eq:ASP1},  GM, REE and LGR of the antisymmetric
basis state are all non-additive if $N\geq3$. Moreover, GM, REE and
LGR of $|\psi_{N-}\rangle^{\otimes2}$ are almost equal to that of
$|\psi_{N-}\rangle$ if $N\gg1$.

Recall that GM, REE and LGR are all equal to $\log N!$ for  the
antisymmetric basis state $|\psi_{N-}\rangle$ according to
\eref{eq:ASP1}, and they are all equal to $N\log N-\log N!$ for the
symmetric basis state $|\psi_{N+}\rangle$ according to
\eref{eq:DickeREEGM}. Since $|\psi_{N+}\rangle$ is non-negative,
theorem~\ref{thm:additivityGnon-negative2} and
proposition~\ref{ProH2} imply that
\begin{eqnarray}
&&R_{\mathrm{L}}(|\psi_{N+}\rangle\otimes|\psi_{N-}\rangle)=E_{\mathrm{R}}(|\psi_{N+}\rangle\otimes|\psi_{N-}\rangle)
=G(|\psi_{N+}\rangle\otimes|\psi_{N-}\rangle)\nonumber\\&&
=G(|\psi_{N+}\rangle)+G(|\psi_{N-}\rangle)=N\log N.
\end{eqnarray}
Surprisingly, GM, REE and LGR are all equal to $N\log N$ for both
$|\psi_{N-}\rangle^{\otimes2}$ and
$|\psi_{N+}\rangle\otimes|\psi_{N-}\rangle$. It is not known whether
this is just a coincidence, or there is a deep reason.

When $N=2$, $\rho_{d,N}$  is an antisymmetric  Werner state, and
\eref{eq:ASP2} reduces to
\begin{eqnarray}
&&R_{\mathrm{L}}\bigl(\rho_{d,2}^{\otimes2}\bigr)=E_{\mathrm{R}}\bigl(\rho_{d,2}^{\otimes2}\bigr)
=G\bigl(\rho_{d,2}^{\otimes2}\bigr)-2\log\frac{d(d-1)}{2}=\log\frac{2d}{d-1}.
\end{eqnarray}
REE of two copies of antisymmetric Werner states was derived by
Vollbrecht and Werner \cite{vw01}, who discovered the Werner state
as the first counterexample to the additivity of REE. In the case of
two-qutrit antisymmetric Werner state, the non-additivity of the
three measures is in contrast with the additivity of entanglement of
formation \cite{Yura}.

\Eref{eq:ASP2} can also be generalized to the tensor product of two
antisymmetric projector states whose respective single-particle
Hilbert spaces have different dimensions, say $d_1, d_2$,
respectively. Suppose $N\leq d_1\leq d_2$, with a similar reasoning
that leads to \eref{eq:ASP2}, one can show that
\begin{eqnarray}
&&R_{\mathrm{L}}\bigl(\rho_{d_1,N}\otimes\rho_{d_2,N}\bigr)=E_{\mathrm{R}}\bigl(\rho_{d_1,N}\otimes\rho_{d_2,N}\bigr)\nonumber\\
&&=G\bigl(\rho_{d_1,N}\otimes\rho_{d_2,N}\bigr)-\log\bigl[\mathrm{tr}\bigl(P_{d_1,N}\bigl)\bigr]
-\log\bigl[\mathrm{tr}\bigl(P_{d_2,N}\bigl)\bigr]\nonumber\\
&&=\log\frac{d_1^NN!(d_1-N)!}{d_1!}.
\end{eqnarray}
Interestingly, REE and LGR of $\rho_{d_1,N}\otimes\rho_{d_2,N}$ are
independent of $d_2$, as long as $N\leq d_1\leq d_2$.

The antisymmetric projector state can be seen as a uniform mixture
of Slater determinant states. The above results on REE and LGR can
also be generalized to an arbitrary mixture of Slater determinant
states. Let $\rho_N=\sum_j p_j|\psi_j\rangle\langle\psi_j|$, where
$|\psi_j\rangle$s are $N$-partite Slater determinant states, and
$\{p_j\}$ is a probability distribution. Due to the convexity of REE
and LGR,
\begin{eqnarray}\label{eq:MSlater1}
E_{\mathrm{R}}(\rho_N^{\otimes n})\leq
E_{\mathrm{R}}(|\psi_{N-}\rangle^{\otimes n}),\qquad
R_{\mathrm{L}}(\rho_N^{\otimes n})\leq
R_{\mathrm{L}}(|\psi_{N-}\rangle^{\otimes n}).
\end{eqnarray}
On the other hand, since $\rho_N$ can be turned into the
antisymmetric projector state by twirling,
\begin{eqnarray}\label{eq:MSlater2}
E_{\mathrm{R}}(\rho_N^{\otimes n})\geq
E_{\mathrm{R}}(\rho_{d,N}^{\otimes n}),\qquad
R_{\mathrm{L}}(\rho_N^{\otimes n})\geq
R_{\mathrm{L}}(\rho_{d,N}^{\otimes n}).
\end{eqnarray}
Combining \eref{eq:MSlater1}, \eref{eq:MSlater2} and
proposition~\ref{pro:ASP2}, we obtain
\begin{proposition}
REE and LGR of  any convex mixture $\rho_N$ of $N$-partite Slater
determinant states  satisfy the following equations:
\begin{eqnarray}
  & E_{\mathrm{R}}(\rho_N) = R_{\mathrm{L}}(\rho_N)= \log N!, &\nonumber\\
  &\log\frac{d^NN!(d-N)!}{d!}\leq
   E_{\mathrm{R}}(\rho_N^{\otimes2}) \leq
  R_{\mathrm{L}}(\rho_N^{\otimes2}) \leq N\log N.&
\end{eqnarray}
If $N\geq 3$, GM, REE and LGR of any convex mixture of Slater
determinant states are all non-additive.
\end{proposition}

\subsection{\label{sec:ass4} Generalized antisymmetric states}
In all  counterexamples to the additivity of GM considered so far,
the dimension of the single-particle Hilbert space is at least
three. However, this constraint is not necessary. We shall
demonstrate this point with the example of generalized antisymmetric
states.

Let $d,p,k$ be three integers satisfying $k\leq d^p$ and $\phi$ the
function from $1,2,\ldots, d^p$ to $p$-tuples defined as follows,
\begin{eqnarray}
\phi(1)&=&(0,0,\ldots,0,0), \nonumber\\
\phi(2)&=&(0,0,\ldots,0,1),\nonumber\\
&&\vdots \nonumber\\
\phi(d^p)&=&(d-1,d-1,\ldots,d-1,d-1).
\end{eqnarray}
For each triple $d, p, k$, define an $N$-partite state with $N=kp$
as follows,
\begin{eqnarray}
|\psi_{d,p,k}\rangle:=\frac{1}{\sqrt{k!}}\sum_{\sigma}\mathrm{sgn}(\sigma)|\phi(\sigma(1)),\ldots,\phi(\sigma(k))\rangle.
\end{eqnarray}
$|\psi_{d,p,k}\rangle$ can be seen as a $k$-partite antisymmetric
basis state with single-particle Hilbert space of dimension $d^p$,
if we divide the $kp$ parties into $k$ blocks each with $p$ parties,
and view each block as a single party. The state $|\psi_{d, p,
d^p}\rangle$ is exactly the generalized antisymmetric state
introduced by Bravyi \cite{bra03}.

By definition, $\Lambda^2(|\psi_{d,p,k}\rangle)\leq
\Lambda^2(|\psi_{k-}\rangle)=\frac{1}{k!}$, and since
$|\langle\psi_{d,p,k}|\phi(1),\ldots,\phi(k)\rangle|^2=\frac{1}{k!}$,
we have
\begin{equation}
G(|\psi_{d,p,k}\rangle)= G(|\psi_{k-}\rangle)=\log (k!).
\label{eq:GM_GA}
\end{equation}
When $k=d^p$, this result reduces to that found by Bravyi
\cite{bra03}.

To compute REE and LGR of $|\psi_{d,p,k}\rangle$, note that the
separable state
\begin{eqnarray}
&&\rho_{\mathrm{sep}}=\frac{1}{k!}\sum_{\sigma\in
S_k}|\phi(\sigma(1)),\ldots,\phi(\sigma(k))\rangle\langle\phi(\sigma(1)),\ldots,\phi(\sigma(k))|
\end{eqnarray}
 can be written in the  form,
$\rho_{\mathrm{sep}}=\frac{1}{k!}|\psi_{d,p,k}\rangle\langle\psi_{d,p,k}|+\rho_{\perp}$,
where  $\rho_{\perp}$ (not normalized) is orthogonal to
$|\psi_{d,p,k}\rangle\langle\psi_{d,p,k}|$; hence,
$R_{\mathrm{L}}(|\psi_{d,p,k}\rangle)\leq\log(k!)$  according to
\eref{eq:GR} and \eref{eq:LGR}. On the other hand,
$R_{\mathrm{L}}(|\psi_{d,p,k}\rangle)\geq
E_{\mathrm{R}}(|\psi_{d,p,k}\rangle)\geq
G(|\psi_{d,p,k}\rangle)=\log(k!)$, so we obtain
\begin{eqnarray}
R_{\mathrm{L}}(|\psi_{d,p,k}\rangle)=E_{\mathrm{R}}(|\psi_{d,p,k}\rangle)=\log(k!).
\end{eqnarray}
When $k=d^p$, this result reduces to that found by Wei {\it et al}
\cite{weg04}.

Now consider two copies of generalized antisymmetric states. For the
same reason as in the case of single copy,
$\Lambda^2(|\psi_{d,p,d^p}\rangle^{\otimes2})\leq
\Lambda^2(|\psi_{d^p-}\rangle^{\otimes2})=(d^p)^{-d^p} $. Suppose
$|a_j\rangle$ is a $d\otimes d$ maximally entangled state across the
cut $A_j^1:A_j^2$, for $j=1,\ldots,p$, then $
\otimes_{j=1}^p|a_j\rangle$ is a maximally entangled state across
the cut $A_1^1\ldots A_p^1:A_1^2\ldots A_p^2$. According to
proposition~\ref{pro:ASP2} in \sref{sec:ass3}, we obtain
\begin{equation}
\Lambda^2(|\psi_{d,p,d^p}\rangle^{\otimes2})=(d^p)^{-d^p},\qquad
G(|\psi_{d,p,d^p}\rangle^{\otimes2})=d^p\log(d^p).
\end{equation}
GM of generalized antisymmetric states are also non-additive when
$d^p>2$, the eight-qubit state $|\psi_{2,2,4}\rangle$ being one of
such examples. It is interesting to know if there exists a
multiqubit state with fewer number of parties whose GM is
non-additive.

When $k<d^p$, it is  not so easy to compute GM of
$|\psi_{d,p,k}\rangle^{\otimes2}$. Nevertheless, a good upper bound
is enough to reveal the non-additivity of GM in many cases. For
example, let
$|\ph\rangle=\bigl(\frac{1}{\sqrt{d}}\sum_{j=0}^{d-1}|jj\rangle\bigr)^{\otimes
kp}$, then
$\Lambda^2(|\psi_{d,p,k}\rangle^{\otimes2})\geq|\langle\psi_{d,p,k}^{\otimes
2}|\ph\rangle|^2=\frac{1}{d^{kp}}$ and thus
$G(|\psi_{d,p,k}\rangle^{\otimes2})\leq kp\log d$. In conjunction
with \eref{eq:GM_GA}, we can discover many multiqubit states whose
GM is non-additive, such as $|\psi_{2,3,6}\rangle$,
$|\psi_{2,3,7}\rangle$ and $|\psi_{2,3,8}\rangle$.

\section{\label{sec:NonAddG} Non-additivity of geometric measure
 of generic multipartite  states}

Many examples and counterexamples to the additivity of GM presented
in the previous sections invite the following question: What is the
typical behavior concerning the additivity property of GM of
multipartite states, additive or non-additive? In this section, we
show that if the number of parties is sufficiently large, and the
dimensions of the local Hilbert spaces are comparable, then GM is
not strong additive for almost all pure multipartite states. What's
more surprising, for generic pure states with real entries in the
computational basis,  GM for one copy and two copies, respectively,
are almost equal. This conclusion follows from the following two
observations which are of independent interest: First, almost all
multipartite pure states are nearly maximally entangled with respect
to  GM and REE; second, there is a nontrivial universal upper bound
for GM of multipartite states with tensor-product form. Our results
have  significant implications for universal one-way quantum
computation and to asymptotic state transformation under LOCC.

\subsection{\label{sec:UUB}Universal upper bound for the geometric measure of multipartite states with tensor-product form}
In this section, we derive a universal upper bound for GM of the
tensor product of two multipartite states, and discuss its
implications.
\begin{proposition}\label{pro:upperboundU}
Suppose $\rho_N$ and $\rho_N^\prime$ are two $N$-partite states on
the Hilbert space $\bigotimes^N_{j=1} \mathcal{H}_j$ with
$\mathrm{Dim}~{\mathcal{H}_j}=d_j$; define $d_T=\prod_{j=1}^Nd_j$;
then  $G(\rho_N\otimes\rho_N^\prime)\leq \log
d_T-\log\mathrm{tr}(\rho_N\rho_N^{\prime*})$. In particular,
$G(\rho_N\otimes\rho_N^*)\leq \log d_T-\log\mathrm{tr}(\rho_N^2)$;
$G^\infty(\rho_N)\leq \frac{1}{2}G(\rho_N^{\otimes2})\leq
\frac{1}{2}\log d_T-\frac{1}{2}\log\mathrm{tr}(\rho_N\rho_N^*)$.
\end{proposition}
\begin{proof}
Let $|\varphi_N\rangle=\bigotimes_{j=1}^N |\phi_{d_j}^+\rangle$,
where
$|\phi_{d_j}^+\rangle=\frac{1}{\sqrt{d_j}}\sum_{k=0}^{d_j-1}|kk\rangle$
is a maximally entangled state (also a pure isotropic state) across
the two copies of the $j$th party. According to \eref{eq:App:innerG}
in the Appendix, $\langle
\varphi_N|\rho_N\otimes\rho_N^\prime|\varphi_N\rangle=\mathrm{tr}(\rho_N\rho_N^{\prime*})/d_T$;
hence, $\Lambda^2(\rho_N\otimes
\rho_N^\prime)\geq\mathrm{tr}(\rho_N\rho_N^{\prime*})/d_T$ and thus
$G(\rho_N\otimes \rho_N^\prime)\leq \log
d_T-\log\mathrm{tr}(\rho_N\rho_N^{\prime*})$.
\end{proof}

If GM of $\rho_N$ is strong additive and thus
$G(\rho_N\otimes\rho_N^*)=G(\rho_N)+G(\rho_N^*)=2G(\rho_N)$,
proposition~\ref{pro:upperboundU} implies that $G(\rho_N)\leq
\frac{1}{2}\log d_T-\frac{1}{2}\log\mathrm{tr}(\rho_N^2)$.  In other
words, GM cannot be strong additive if the states  are too entangled
with respect to GM. This intuition will be made more rigorous in
theorem~\ref{thm:GMNonAdd}.
 For states with real entries in the
computational basis (real states for short),
proposition~\ref{pro:upperboundU} sets a universal upper bound for
$G(\rho_N^{\otimes 2})$ and $G^\infty(\rho_N)$, that is,
$G^\infty(\rho_N)\leq \frac{1}{2}G(\rho_N^{\otimes2})\leq
\frac{1}{2}\log d_T-\frac{1}{2}\log\mathrm{tr}(\rho_N^2)$. According
to a similar reasoning as in the proof of
proposition~\ref{pro:upperboundU}, the same upper bound also applies
to any state that is equivalent to its complex conjugate under local
unitary transformations. Hence, GM cannot be  additive  if such
states are too entangled with respect to GM.

If $d_j=d, \forall j$ and $d\geq N$, the universal upper bound for
$G(\rho_N^{\otimes 2})$ of real states $\rho_N$ given in
proposition~\ref{pro:upperboundU} is saturated for the antisymmetric
projector states (see \sref{sec:ass3}). If  $d_j=d, \forall j$ and
$N$ is even, there is a simple scheme for constructing a pure state
whose GM saturates the upper bound: Divide the parties into
$\frac{N}{2}$ pairs, and choose a maximally entangled state for each
pair of parties, then the tensor product of the $\frac{N}{2}$
maximally entangled states (note that all the entries of the state
can be made real by a suitable local unitary transformation) is such
a candidate. Moreover, GM of the state so constructed is additive,
so are REE and LGR. A more attractive example which saturates the
upper bound is the cluster state with even number of qubits, whose
GM, REE and LGR are all equal to $N/2$ and are additive
\cite{mmv07}. Hence, proposition~\ref{pro:upperboundU} implies that,
in any multipartite Hilbert space with even number of parties and
equal local dimension, any pure state with real entries in the
computational basis cannot be more entangled with respect to AGM
than the tensor product of bipartite maximally entangled states, or
the cluster state, for a multiqubit system.

If $N$ is odd, however, there may exists no pure state (even with
complex entries in the computational basis) that can saturate the
upper bound given in proposition~\ref{pro:upperboundU}. For example,
W state has been shown to be the maximally entangled state with
respect to GM among pure three-qubit states \cite{cxz10,twp09},
while its AGM, which equals to its GM $\log\frac{9}{4}$, is strictly
smaller than the upper bound $\frac{3}{2}\log2$ given in the
proposition.

It is interesting to know whether the same bound is true for states
with arbitrary entries and whether there is a similar universal
upper bound for REE and LGR; in particular, whether AREE or ALGR is
 upper bounded by $\frac{1}{2}\log d_T$. It is also not clear
whether REE and LGR are not strong additive for generic multipartite
states. We have shown in \sref{sec:ass3} that AREE  is upper bounded
by $\frac{1}{2}\log d_T$ for antisymmetric basis states. The same is
true for all symmetric basis states, according to \eref{DickeAREE}.
However, a complete picture is still missing. We hope that our
results can stimulate more progress along this direction.

\subsection{\label{sec:NonAddThm}Non-additivity theorem for  geometric measure of generic multipartite states: a statistical approach}
In this section we prove the following theorem.
\begin{theorem}
\label{thm:GMNonAdd} Suppose pure states  are drawn according to the
Haar measure from the Hilbert space $\bigotimes^N_{j=1}
\mathcal{H}_j$ with $N\geq3$ and $\mathrm{Dim}~{\mathcal{H}_j}=d_j$
($d_j\geq 2, \forall j$); define $d_T=\prod_{j=1}^N d_j$ and
$d_S=\sum_{j=1}^N d_j$. The fraction of pure states whose GM is
strong additive is smaller than
$\exp\bigl[-\frac{2}{3}\sqrt{d_T}+d_S\ln(59Nd_T)\bigr]$; the
fraction of pure states $|\psi\rangle$ such that  $\bigl[\log
d_T-\log(d_S\ln d_T)-\log\frac{9}{2}\bigr]\leq  G(|\psi\rangle)<
G(|\psi\rangle\otimes|\psi^*\rangle )\leq \log d_T$ is larger than
$1-d_T^{-d_S}$. For pure states with real entries in the
computational basis, the fraction of pure states whose GM is
additive is smaller than
$\exp\bigl[-\frac{1}{3}\sqrt{d_T}+d_S\ln(59Nd_T)\bigr]$; the
fraction of pure states $|\psi\rangle$ such that  $\bigl[\log
d_T-\log(d_S\ln d_T)-\log9\bigr]\leq  G(|\psi\rangle)<
G(|\psi\rangle^{\otimes 2})\leq \log d_T$ is larger than
$1-d_T^{-d_S}$.
\end{theorem}
Theorem~\ref{thm:GMNonAdd} implies that  GM is not strong additive
for almost all pure multipartite states, if the number of parties is
sufficiently large, and the dimensions of the local Hilbert spaces
 are comparable. Moreover, GM of $|\psi\rangle$ and $|\psi\rangle\otimes|\psi^*\rangle$, respectively, is almost equal. If the dimensions
of the local Hilbert spaces are equal,  the probability that GM is
strong additive decreases doubly exponentially with the number of
parties $N$. Concerning real states, GM is non-additive for almost
all pure multipartite states, and GM of one copy and two copies,
respectively, is almost equal. The generalization to mixed states is
immediate, since GM of any mixed state is equal to GM of its
purification \cite{jung08} (see also \eref{eq:ReduceGM}).

Theorem~\ref{thm:GMNonAdd} is an immediate consequence of
proposition~\ref{pro:upperboundU} in \sref{sec:UUB} and proposition
\ref{pro:GMME} presented below. The later proposition, which is
inspired by a similar result on multiqubit pure states of
\cite{GFE09}, shows that almost all multipartite pure states are
nearly maximally entangled with respect to GM.

\begin{proposition}\label{pro:GMME}
Suppose pure states  are drawn according to the Haar measure from
the Hilbert space $\bigotimes^N_{j=1} \mathcal{H}_j$ with $N\geq3$
and $\mathrm{Dim}~{\mathcal{H}_j}=d_j$ ($d_j\geq 2, \forall j$);
define $d_T=\prod_{j=1}^N d_j$ and $d_S=\sum_{j=1}^N d_j$. The
fraction of pure states whose GM is smaller than $\log
d_T-\log(d_S\ln d_T)-\log\frac{9}{2}$ is less than $d_T^{-d_S}$; the
fraction of pure states whose GM is smaller than $\frac{1}{2}\log
d_T$ is less than
$\exp\bigl[-\frac{2}{3}\sqrt{d_T}+d_S\ln(59Nd_T)\bigr]$. For pure
states with real entries in the computational basis, the fraction of
pure states whose GM is smaller than $\log d_T-\log(d_S\ln
d_T)-\log9$ is less than $d_T^{-d_S}$; the fraction of pure states
whose GM is smaller than $\frac{1}{2}\log d_T$ is less than
$\exp\bigl[-\frac{1}{3}\sqrt{d_T}+d_S\ln(59Nd_T)\bigr]$.
\end{proposition}
By means  of the relation among the three measures GM, REE, and LGR
(see \eref{eq:REE-GM}), we obtain
\begin{corollary}\label{cor:REEME}
Suppose pure states  are drawn according to the Haar measure from
the Hilbert space $\bigotimes^N_{j=1} \mathcal{H}_j$ with $N\geq3$
and $\mathrm{Dim}~{\mathcal{H}_j}=d_j$ ($d_j\geq 2, \forall j$);
define $d_T=\prod_{j=1}^N d_j$ and $d_S=\sum_{j=1}^N d_j$. The
fraction of pure states whose REE or LGR is smaller than $\log
d_T-\log(d_S\ln d_T)-\log\frac{9}{2}$ is less than $d_T^{-d_S}$; the
fraction of pure states whose REE or LGR is smaller than
$\frac{1}{2}\log d_T$ is less than
$\exp\bigl[-\frac{2}{3}\sqrt{d_T}+d_S\ln(59Nd_T)\bigr]$.  For pure
states with real entries in the computational basis, the fraction of
pure states whose REE or LGR is smaller than $\log d_T-\log(d_S\ln
d_T)-\log9$ is less than $d_T^{-d_S}$; the fraction of pure states
whose REE or LGR  is smaller than $\frac{1}{2}\log d_T$ is less than
$\exp\bigl[-\frac{1}{3}\sqrt{d_T}+d_S\ln(59Nd_T)\bigr]$.
\end{corollary}

Note that  $G(\rho)\leq E_{\mathrm{R}}(\rho)\leq \log d_T$ for any
pure state $\rho$, since $S(\rho\|I/d_T)=\log d_T$.
Proposition~\ref{pro:GMME} and corollary~\ref{cor:REEME} implies
that  almost all multipartite pure states are nearly maximally
entangled with respect to GM and REE, if the number of parties is
sufficiently large, and the dimensions of the local Hilbert spaces
are comparable. In particular, if the dimensions of the local
Hilbert spaces are equal, then the probability that GM (REE, LGR) is
smaller than $\log d_T-\log(d_S\ln d_T)-\log\frac{9}{2}$ decreases
exponentially with the number of parties $N$, and the probability
that GM (REE, LGR) is smaller than $\frac{1}{2}\log d_T$ decreases
doubly exponentially.

\begin{proof}
To prove the proposition, we need the concept of $\varepsilon$-net.
An $\varepsilon$-net $\mathcal{N}_{\varepsilon,N}$ on the set of
pure product states is a set of states that satisfy
\begin{eqnarray}
\max_{|\varphi\rangle\in
\mathrm{PRO}}\min_{|\tilde{\varphi}\rangle\in\mathcal{N}_{\varepsilon,N}
}\Bigl\||\varphi\rangle\langle\varphi|-|\tilde{\varphi}\rangle\langle\tilde{\varphi}|\Bigr\|_1\leq
\varepsilon,
\end{eqnarray}
or equivalently,
\begin{eqnarray}
\min_{|\varphi\rangle\in
\mathrm{PRO}}\max_{|\tilde{\varphi}\rangle\in\mathcal{N}_{\varepsilon,N}
}|\langle\varphi|\tilde{\varphi}\rangle|^2\geq
1-\frac{\varepsilon^2}{4}.
\end{eqnarray}
We shall show that there exists an $\varepsilon$-net with
$|\mathcal{N}_{\varepsilon,N}|\leq(5\sqrt{N}/\varepsilon)^{2d_S}$,
where $|\mathcal{N}_{\varepsilon,N}|$ denotes the number of elements
in the $\varepsilon$-net. From \cite{HLSW04}, we know that there is
an $\varepsilon$-net $\mathcal{M}$ on the Hilbert space of single
qudit with $|\mathcal{M}|\leq(5/\varepsilon)^{2d}$. Let
$\mathcal{M}_j$ be an $(\varepsilon/\sqrt{N})$-net on
$\mathcal{H}_j$ with
$|\mathcal{M}_j|\leq(5\sqrt{N}/\varepsilon)^{2d_j}$ for
$j=1,\ldots,N$, and
$\mathcal{N}_{\varepsilon,N}:=\bigl\{\bigotimes_{j=1}^{N}
|\tilde{a}_j\rangle: |\tilde{a}_j\rangle\in \mathcal{M}_j\bigr\}$.
Suppose $|\varphi\rangle=\bigotimes_{j=1}^N |a_j\rangle$ is an
arbitrary product state,  by definition of the
$(\varepsilon/\sqrt{N})$-net, for each $j$, there exists
$|\tilde{a}_j\rangle\in \mathcal{M}_j$ such that $|\langle
a_j|\tilde{a}_j\rangle|^2\geq 1-\varepsilon^2/4N$. It follows that
the following relation holds for
$|\tilde{\varphi}\rangle=\bigotimes_{j=1}^N|\tilde{a}_j\rangle\in
\mathcal{N}_{\varepsilon,N}$,
\begin{eqnarray}
|\langle\varphi|\tilde{\varphi}\rangle|^2=\Bigl(1-\frac{\varepsilon^2}{4N}\Bigr)^N\geq
1-\frac{\varepsilon^2}{4}.
\end{eqnarray}
Hence, $\mathcal{N}_{\varepsilon,N}$ is an $\varepsilon$-net on the
set of product states with
$|\mathcal{N}_{\varepsilon,N}|\leq(5\sqrt{N}/\varepsilon)^{2d_S}$.

Another ingredient in our proof is the following result on the
concentration of measure: Let $|\Phi\rangle$ be any given pure
state, and $|\Psi\rangle$ be chosen according to the Haar measure,
then we have (assuming $0<\varepsilon<1$)
\begin{eqnarray} \label{eq:concentrationG}
\mathrm{Prob}\{|\langle\Phi|\Psi\rangle|^2\geq
\varepsilon\}=(1-\varepsilon)^{d_T-1}<\exp[-(d_T-1)\varepsilon].
\end{eqnarray}
Since
\begin{eqnarray}
&&\Bigl|
|\langle\varphi|\Psi\rangle|^2-|\langle\tilde{\varphi}|\Psi\rangle|^2\Bigr|
=\Bigl|\mathrm{tr}\Bigl[(|\varphi\rangle\langle\varphi|-|\tilde{\varphi}\rangle\langle\tilde{\varphi}|)|\Psi\rangle\langle\Psi|\Bigr]\Bigr|\nonumber\\
&\leq&\Bigl\||\varphi\rangle\langle\varphi|-|\tilde{\varphi}\rangle\langle\tilde{\varphi}|\Bigr\|_\infty
=\frac{1}{2}\Bigl\||\varphi\rangle\langle\varphi|-|\tilde{\varphi}\rangle\langle\tilde{\varphi}|\Bigr\|_1\leq\frac{\varepsilon}{2},
\end{eqnarray}
the probability that $G(|\Psi\rangle)\leq
-\log\bigl(\frac{3}{2}\varepsilon\bigr)$ is at
most\begin{eqnarray}\label{eq:ProbG1}
&&\fl\mathrm{Prob}\Bigl\{\max_{|\tilde{\varphi}\rangle\in
\mathcal{N}_{\varepsilon,N}}
|\langle\tilde{\varphi}|\Psi\rangle|^2\geq \varepsilon \Bigr\}<
\exp[-(d_T-1)\varepsilon]\bigl|\mathcal{N}_{\varepsilon,N}\bigr|
\leq\exp[-(d_T-1)\varepsilon]\Bigl(\frac{5\sqrt{N}}{\varepsilon}\Bigr)^{2d_S}.\nonumber\\
\end{eqnarray}

Now let $\varepsilon=3d_S\ln d_T/d_T$, the probability that
$G(|\Psi\rangle)\leq [\log d_T-\log(d_S\ln d_T)-\log\frac{9}{2}]$ is
at most\begin{eqnarray}\label{eq:ProbG2}
&&\fl\mathrm{Prob}\Bigl\{\max_{|\tilde{\varphi}\rangle\in
\mathcal{N}_{\varepsilon,N}}
|\langle\tilde{\varphi}|\Psi\rangle|^2\geq\frac{3d_S\ln
d_T}{d_T}\Bigr\}< \exp\Bigl[-d_S\ln d_T+d_S\Bigl(\frac{3\ln
d_T}{d_T}+2\ln\frac{5\sqrt{N}}{3d_S\ln d_T}\Bigr)\Bigr]\nonumber\\
&&<\exp(-d_S\ln d_T)=d_T^{-d_S}.
\end{eqnarray}
Here the last inequality can be derived as follows. Under our
assumption $N\geq 3$, $d_j\geq2, \forall j$ ($d_S\geq 6, d_T\geq8$),
the maximum of $3\ln d_T/d_T+2\ln\bigl[5\sqrt{N}/(3d_S\ln
d_T)\bigr]$ for each $N$, which  is obtained when $d_j=2, \forall
j$, decreases monotonically with $N$. Hence, the global maximum of
$3\ln d_T/d_T+2\ln\bigl[5\sqrt{N}/(3d_S\ln d_T)\bigr]$ is obtained
when $N=3$ and $d_1=d_2=d_3=2$. Since this maximum is negative, the
inequality follows.

Next, let $\varepsilon=\frac{2}{3}d_T^{-1/2}$, the probability that
$G(|\Psi\rangle)\leq \frac{1}{2}\log d_T$ is at most
\begin{eqnarray}\label{eq:ProbG3}
&&\fl \mathrm{Prob}\Bigl\{\max_{|\tilde{\varphi}\rangle\in
\mathcal{N}_{\varepsilon,N}}
|\langle\tilde{\varphi}|\Psi\rangle|^2\geq \frac{2}{3\sqrt{d_T}}
\Bigr\}\leq \exp\Bigl[-\frac{2}{3}\sqrt{d_T}+\frac{2}{3\sqrt{d_T}}+d_S\ln\Bigl(\frac{225}{4}Nd_T\Bigr)\Bigr]\nonumber\\
&&= \exp\Bigl[-\frac{2}{3}\sqrt{d_T}+d_S\ln(59Nd_T)\Bigr]
\exp\Bigl[\frac{2}{3}d_T^{-1/2}+d_S\ln\Bigl(\frac{225}{236}\Bigr)\Bigr]\nonumber\\
&&<\exp\Bigl[-\frac{2}{3}\sqrt{d_T}+d_S\ln(59Nd_T)\Bigr].
\end{eqnarray}
Here the last inequality  can be seen as follows. Under our
assumption $N\geq 3$, $d_j\geq2, \forall j$ ($d_S\geq 6, d_T\geq8$),
the maximum of
$\frac{2}{3}d_T^{-1/2}+d_S\ln\bigl(\frac{225}{236}\bigr)$, which is
obtained when $N=3$ and $d_1=d_2=d_3=2$ ($d_S=6, d_T=8$), is
negative.

Now suppose $|\Psi_{\mathrm{R}}\rangle$ is  chosen according to the
Haar measure from kets with real amplitudes in the computational
basis, and $|\Phi\rangle$ is still a given pure state not
necessarily with real amplitudes. Then we have (assuming
$0<\varepsilon<1$)
\begin{eqnarray}\label{eq:concentrationR}
\mathrm{Prob}\{|\langle\Phi|\Psi_{\mathrm{R}}\rangle|^2\geq
\varepsilon\}\leq(1-\varepsilon)^{\frac{d_T}{2}-1}<\exp\Bigl[-\Bigl(\frac{d_T}{2}-1\Bigr)\varepsilon\Bigr];
\end{eqnarray}
the main difference of the above equation from
(\ref{eq:concentrationG}) is the replacement of $d_T$ by $d_T/2$.
\Eref{eq:concentrationR} can be derived as follows. By a suitable
orthogonal transformation if necessary, $|\Phi\rangle$ can be turned
into the form $(a+b\mathrm{i})|0\rangle+c\mathrm{i}|1\rangle$, where
$|0\rangle$ and $|1\rangle$ are two basis kets within the
orthonormal basis $|0\rangle, |1\rangle, \ldots, |d_T-1\rangle$, and
$a,b,c$ are real numbers satisfying $a^2+b^2+c^2=1$. Suppose
$|\Psi_{\mathrm{R}}\rangle=\sum_{j=0}^{d_T-1}x_j|j\rangle$, where
$x_j$s are real numbers satisfying $\sum_{j=0}^{d_T-1} x_j^2=1$;
then $|\langle\Phi|\Psi_{\mathrm{R}}\rangle|^2\leq x_0^2+x_1^2$.
Hence
\begin{eqnarray}
\fl \mathrm{Prob}\{|\langle\Phi|\Psi_{\mathrm{R}}\rangle|^2\geq
\varepsilon\}\leq \mathrm{Prob}\Biggl\{x_0^2+x_1^2\geq
\varepsilon\Biggl|\sum_{j=0}^{d_T-1}
x_j^2=1\Biggr\}=(1-\varepsilon)^{\frac{d_T}{2}-1}.
\end{eqnarray}

According to the same reasoning that leads to \eref{eq:ProbG1}, the
probability that $G(|\Psi_{\mathrm{R}}\rangle)\leq
-\log\bigl(\frac{3}{2}\varepsilon\bigr)$ is at most
\begin{eqnarray}\label{eq:ProbR1}
&&\fl\mathrm{Prob}\Bigl\{\max_{|\tilde{\varphi}\rangle\in
\mathcal{N}_{\varepsilon,N}}
|\langle\tilde{\varphi}|\Psi_{\mathrm{R}}\rangle|^2\geq \varepsilon
\Bigr\}<\exp[-\Bigl(\frac{d_T}{2}-1\Bigr)\varepsilon]\Bigl(\frac{5\sqrt{N}}{\varepsilon}\Bigr)^{2d_S}.
\end{eqnarray}
Let $\varepsilon=6d_S\ln d_T/d_T$, the probability that
$G(|\Psi\rangle)\leq [\log d_T-\log(d_S\ln d_T)-\log9]$ is at most
\begin{eqnarray}\label{eq:ProbR2}
&&\fl\mathrm{Prob}\Bigl\{\max_{|\tilde{\varphi}\rangle\in
\mathcal{N}_{\varepsilon,N}}
|\langle\tilde{\varphi}|\Psi\rangle|^2\geq\frac{6d_S\ln
d_T}{d_T}\Bigr\}< \exp\Bigl[-d_S\ln d_T+2d_S\Bigl(\frac{3\ln
d_T}{d_T}+\ln\frac{5\sqrt{N}}{6d_S\ln d_T}\Bigr)\Bigr]\nonumber\\
&&<\exp(-d_S\ln d_T)=d_T^{-d_S}.
\end{eqnarray}
Next, let $\varepsilon=\frac{2}{3}d_T^{-1/2}$, the probability that
$G(|\Psi\rangle)\leq \frac{1}{2}\log d_T$ is at most
\begin{eqnarray}\label{eq:ProbR3}
&&\fl \mathrm{Prob}\Bigl\{\max_{|\tilde{\varphi}\rangle\in
\mathcal{N}_{\varepsilon,N}}
|\langle\tilde{\varphi}|\Psi\rangle|^2\geq \frac{2}{3\sqrt{d_T}}
\Bigr\}\leq \exp\Bigl[-\frac{1}{3}\sqrt{d_T}+\frac{2}{3\sqrt{d_T}}+d_S\ln\Bigl(\frac{225}{4}Nd_T\Bigr)\Bigr]\nonumber\\
&&<\exp\Bigl[-\frac{1}{3}\sqrt{d_T}+d_S\ln(59Nd_T)\Bigr].
\end{eqnarray}
The derivation of \eref{eq:ProbR2} and \eref{eq:ProbR3} is similar
to that of \eref{eq:ProbG2} and \eref{eq:ProbG3}.
\end{proof}

\subsection{\label{sec:implication}Implications of additivity property for one-way quantum computation
and for asymptotic state transformation} Recently, Gross {\it et al}
\cite{GFE09} (see also \cite{BMW09}) showed that most quantum states
are too entangled to be useful as computational resources. One of
the key ingredient in their proof is the observation that almost all
pure multiqubit  states are nearly maximally entangled with respect
to GM. However, their arguments would break down, if measurements
are allowed on the tensor product of the resource states, since
$\rho\otimes \rho^*$ is just moderate entangled (GM is nearly one
half of the maximal possible value) for a generic pure multiqudit
states $\rho$, according to theorem~\ref{thm:GMNonAdd}. In
particular, two copies of $\rho$ is moderate entangled if $\rho$ is
a real state. Hence, it is conceivable that we may realize universal
quantum computation on certain family of multiqudit states if they
come in pairs, even if this is impossible on a single copy. It would
be very desirable to construct an explicit example of such a family
of multiqudit states or disprove this possibility. However, a
detailed investigation along this direction would well go beyond the
scope of this paper.

Corollary~\ref{cor:REEME} has a significant implication for
asymptotic state transformation. In particular, it implies that
almost all multiqudit pure states cannot be prepared reversibly with
multipartite GHZ  states (of various numbers of parties) under
asymptotic LOCC, unless REE is non-additive for generic multiqudit
states. This can be seen as follows. According to the result of
Linden {\it et al} \cite{lpsw99}, reversible transformation between
two pure states under asymptotic LOCC would mean that the ratio of
the bipartite AREE $E_{\mathrm{R}}^\infty(A_j:\tilde{A}_j)$ to the
$N$-partite AREE $E_{\mathrm{R}}^\infty$ is conserved, for
$j=1,2,\ldots, N$, where $\tilde{A}_j$ denotes all the parties
except $A_j$. As a result, the ratio $\bigl[\sum_{j=1}^N
E_{\mathrm{R}}^\infty(A_j:\tilde{A}_j)\bigr]/E_{\mathrm{R}}^\infty$
is conserved. If a state $|\psi\rangle\langle\psi|$ can be prepared
reversibly with $n_k$ copies of $k$-partite GHZ states for
$k=2,\ldots,N$, then
\begin{eqnarray}\label{eq:AREEratio}
\frac{\sum_{j=1}^N
E_{\mathrm{R}}^\infty\bigl(|\psi\rangle_{A_j:\tilde{A}_j}\bigr)}{E_{\mathrm{R}}^\infty(|\psi\rangle)}=\frac{\sum_{k=2}^N
k n_k}{\sum_{k=2}^N n_k}\geq 2,
\end{eqnarray}
where we have used the fact that REE of the tensor product of GHZ
type states is additive. On the other hand,
$E_{\mathrm{R}}^{\infty}\bigl(|\psi\rangle_{A_j:\tilde{A}_j}\bigr)=E_{\mathrm{R}}\bigl(|\psi\rangle_{A_j:\tilde{A}_j}\bigr)\leq
\log d$, in addition, $E_{\mathrm{R}}(|\psi\rangle)>\frac{1}{2}\log
d_T=\frac{N}{2}\log d$ for almost all multiqudit  pure  states,
according to corollary~\ref{cor:REEME}. If REE of $|\psi\rangle$ is
additive, then we have
\begin{eqnarray}
\frac{\sum_{j=1}^N
E_{\mathrm{R}}^\infty\bigl(|\psi\rangle_{A_j:\tilde{A}_j}\bigr)}{E_{\mathrm{R}}^\infty(|\psi\rangle)}=\frac{\sum_{j=1}^N
E_{\mathrm{R}}\bigl(|\psi\rangle_{A_j:\tilde{A}_j}\bigr)}{E_{\mathrm{R}}(|\psi\rangle)}<
\frac{N\log d}{\frac{N}{2}\log d}= 2,
\end{eqnarray}
which contradicts \eref{eq:AREEratio}. Hence, almost all multiqudit
pure states cannot be prepared reversibly under asymptotic LOCC,
unless  REE is non-additive for generic multiqudit pure states. Our
observation adds to the evidence that  a reversible entanglement
generating set \cite{bpr00,AVC} with a finite cardinality may not
exist.

As a concrete example, similar reasoning has been employed  by
Ishizaka and Plenio \cite{ip05} to  show that $|\psi_{3-}\rangle$
cannot be generated reversibly from the GHZ state and EPR pairs
under asymptotic LOCC if its REE is additive. The same is true for
$|\psi_{N-}\rangle$ with $N\geq3$, since
$E_{\mathrm{R}}(|\psi_{N-}\rangle)=\log (N!)>\frac{N}{2}\log N$
\cite{weg04,wei08,hmm08} (see also \eref{eq:ASP1} in
\sref{sec:ass3}).

\section{Summary}

In this paper, we have studied the additivity property of three main
multipartite entanglement measures, namely GM, REE and LGR.

Firstly, we proved the strong additivity of GM  of  non-negative
states,  thus simplifying the computation of GM and AGM of a large
family of states of  either experimental or theoretical interest.
Thanks to the connection among the three measures, GM of
non-negative states provides a lower bound for AREE and ALGR, and a
new approach for proving the additivity of REE and LGR for states
with certain group symmetry. In particular, we  proved the strong
additivity of GM and the additivity of REE of  Bell diagonal states,
 maximally correlated generalized Bell diagonal states,
generalized Dicke states and their reduced states after tracing out
one party,  the Smolin state and  D\"{u}r's multipartite entangled
states etc. The additivity of LGR of generalized Dicke states and
the Smolin state was also shown. These results can be applied to
studying state discrimination under LOCC \cite{hmm06,hmm08}, the
classical capacity of quantum multi-terminal  channels. The result
on AREE is also useful in studying state transformation  either
under asymptotic LOCC or under asymptotic non-entangling operations.
For non-negative bipartite states, the result on AREE also leads to
a new lower bound for entanglement of formation and entanglement
cost. The result on GM and AGM may  find applications in the study
of quantum channels due to the connection between pure tripartite
states and quantum channels \cite{wh02}.

Secondly, we  established a simple connection between the
permutation symmetry and the additivity property of multipartite
entanglement measures. In particular, we showed that   GM is
non-additive for antisymmetric states shared over three or more
parties. Also, we gave a unified explanation of the non-additivity
of the three measures GM, REE and LGR of the antisymmetric projector
states, and derive analytical formulae of the three measures for one
copy and tow copies of such states.  Our results on antisymmetric
states are
 expected to be useful in the study of fermion systems, which
are described by antisymmetric states due to the super-selection
rule.

Thirdly, we showed  that almost all multipartite pure states are
maximally entangled with respect to GM and REE. However, their GM is
not strong additive; moreover, for generic pure states with real
entries in the computational basis, GM of one copy  and  two copies,
respectively, are almost equal. Based on these observations, we
showed that more states may be suitable for universal quantum
computation, if measurements can be performed on two copies of the
resource states. We also showed that, for almost all multipartite
pure states, the additivity of their REE implies the irreversibility
in generating them from GHZ type states under LOCC, even in the
asymptotic sense.

There are also quite a few open problems  which can be new
directions in the future study of multipartite entanglement.
\begin{enumerate}
  \item Are  GM and REE of arbitrary two-qubit states and pure three-qubit states
  additive?
  \item Are  GM and  REE of arbitrary symmetric states  additive? We cannot find any counterexamples at the moment; however, the
  possibility has not been excluded.

  \item   When are   GM and REE of bipartite mixed antisymmetric  states
  additive or non-additive?

  \item What are  AGM, AREE and ALGR of the antisymmetric projector
  states?  It is enough to compute any one of the three measures, since
  they are related to each other by the simple equalities in proposition~\ref{pro:ASPtensor}.

\item Are GM, REE and LGR non-additive for  generic multipartite
states?

  \item Does there exist a family of quantum states such that two
copies are universal for quantum computation while one copy is not?

\end{enumerate}

\section*{Acknowledgment}
The authors thank Professor Andreas Winter for pointing out the
connection between the non-additivity of GM and the Werner-Holevo
channel, and thank him and Dr. Akimasa Miyake for carefully
introducing the knowledge of one-way quantum computation, which
greatly stimulates the research in \sref{sec:NonAddG}. The authors
also thank Dr. Fernando Guadalupe Santos Lins Brand\~ao and Dr.
Tzu-Chieh Wei for carefully reading the manuscript. LC thanks
Professor Dong Yang for reading the first version of the manuscript.
HZ thanks Dr. Ke Li for pointing out that the proof of
theorem~\ref{thm:additivityGnon-negative2} does not rely on
lemma~\ref{le:closest2}. MH thanks Dr. Masaki Owari for helpful
discussions. The Centre for Quantum Technologies is funded by the
Singapore Ministry of Education and the National Research Foundation
as part of the Research Centres of Excellence programme. MH is also
partially supported by a Grant-in-Aid for Scientific Research in the
Priority Area ``Deepening and Expansion of Statistical Mechanical
Informatics (DEX-SMI),'' No. 18079014 and a MEXT Grant-in-Aid for
Young Scientists (A) No. 20686026.

 \emph{Note added}. Upon completion of this work, we found that Plastino
 {\it et al} \cite{pmd09} had derived a similar criterion as proposition~\ref{pro:Slater} for determining
 when an antisymmetric state is a Slater determinant state.

\section*{Appendix}
Suppose $\rho_N$ and $\rho_N^\prime$ are two $N$-partite states on
the Hilbert space $\bigotimes^N_{j=1} \mathcal{H}_j$ with
$\mathrm{Dim}~{\mathcal{H}_j}=d_j$; define  $d_T=\prod_{j=1}^N d_j$,
\begin{eqnarray}
&&|\varphi_N\rangle=\bigotimes_{j=1}^N |a_{V_{j}}\rangle, \qquad
|a_{V_{j}}\rangle=\frac{1}{\sqrt{d_j}}\sum_{k,l=0}^{d_j-1}V_{j,kl}|kl\rangle,\nonumber\\
&&V_{j}=V_{j,kl}|k\rangle\langle l|,\qquad
\mathrm{tr}(V_jV_j^{\dag})=d_j, \qquad
\mathcal{V}=\bigotimes_{j=1}^NV_j.
\end{eqnarray}
In this appendix, we prove the following formula,
\begin{eqnarray}
\label{eq:App:innerG} &&
\langle\varphi_N|\rho_{N}\otimes\rho_N^\prime|\varphi_N\rangle
=\frac{1}{d_T}\mathrm{tr}\bigl(\rho_{N}^{1/2}\mathcal{V}\rho_N^{\prime*}
\mathcal{V}^{\dag}\rho_{N}^{1/2}\bigl),
\end{eqnarray}
where the complex conjugate is taken in the computational basis. The
formula reduces to
\begin{eqnarray}
\label{eq:App:inner3} &&
\langle\varphi_N|\rho_{N}\otimes\rho_N^\prime|\varphi_N\rangle
=\frac{1}{d^N}\mathrm{tr}\bigl(\rho_{N}^{1/2}V^{\otimes
N}\rho_N^{\prime*} {V^\dag}^{\otimes N}\rho_{N}^{1/2}\bigl).
\end{eqnarray}
in the special case  $d_j=d, V_j=V, \forall j$.

\begin{proof}
\begin{eqnarray}
&& \langle\varphi_N|\rho_{N}\otimes\rho_N^\prime|\varphi_N\rangle
=\Bigl\|\rho_{N}^{1/2}\otimes{\rho_{N}^\prime}^{1/2}|\varphi_N\rangle\Bigr\|^2\nonumber\\
&&=
\frac{1}{d_T}\Biggl\|\sum_{k_1,l_1,\ldots,k_N,l_N}V_{1,k_1l_1}\cdots
V_{N,k_Nl_N}\rho_{N}^{1/2}|k_1,\ldots,k_N\rangle\otimes {\rho_{N}^\prime}^{1/2}|l_1,\ldots,l_N\rangle\Biggr\|^2\nonumber\\
&&=\frac{1}{d_T}
\Biggl\|\sum_{k_1,l_1,\ldots,k_N,l_N}V_{1,k_1l_1}\cdots
V_{N,k_Nl_N}\rho_{N}^{1/2}|k_1,\ldots,k_N\rangle\langle l_1,\ldots,l_N|{\rho_{N}^{\prime T}}^{1/2}\Biggr\|_{\mathrm{HS}}^2\nonumber\\
&& =\frac{1}{d_T}\Biggl\|\rho_{N}^{1/2}\Biggl(\bigotimes_{j=1}^N
V_j\Biggr){\rho_{N}^{\prime *}}^{1/2}
\Biggr\|_{\mathrm{HS}}^2=\frac{1}{d_T}\Bigl\|\rho_{N}^{1/2}\mathcal{V}{\rho_{N}^{\prime*}}^{1/2}
\Bigr\|_{\mathrm{HS}}^2\nonumber\\
&&=\frac{1}{d_T}\mathrm{tr}\Bigl(\rho_{N}^{1/2}\mathcal{V}\rho_{N}^{\prime
*} {\mathcal{V}^\dag}\rho_{N}^{1/2}\Bigl).\nonumber
\end{eqnarray}
\end{proof}

\end{document}